%% file: arXiv.tex
\documentclass[superscriptaddress, 10pt, amsmath,amssymb, aps, pra, twocolumn,preprintnumbers]{revtex4-2}

\usepackage[utf8]{inputenc}
\usepackage[T1]{fontenc}
\usepackage{amsmath}
\usepackage{amsfonts}
\usepackage{amssymb}
\usepackage{amsthm}
\usepackage{dsfont}
\usepackage{graphicx}
\usepackage[dvipsnames]{xcolor}
\usepackage{enumitem}
\usepackage{hyperref}
\usepackage{physics}
\usepackage{bbm}
\usepackage{indentfirst}
\usepackage{thmtools}
\usepackage{thm-restate}
\usepackage{soul}
\usepackage{comment}
\setstcolor{red}
\usepackage{breakurl}
\usepackage{mathtools}
\usepackage{titletoc}
\usepackage{chngcntr}
\usepackage[normalem]{ulem}

\usepackage{titletoc}

\titlecontents{section}
  [1.5em]
  {}
  {\contentslabel{2em}}
  {}
  {\hfill\contentspage}

\titlecontents{subsection}
  [3.5em]
  {}
  {\contentslabel{2.5em}}
  {}
  {\hfill\contentspage}

\usepackage{bibunits}
\defaultbibliographystyle{aps}
\defaultbibliography{refs}

\usepackage{algorithm,algorithmic}

\usepackage{setspace}

\newtheorem{theorem}{Theorem}

\newtheorem{definition}{Definition}

\newtheorem{lemma}{Lemma}

\newtheorem{proposition}{Proposition}

\makeatletter
\newcommand{\bigboxplus}{\mathop{\vphantom{\sum}\mathchoice
  {\vcenter{\hbox{\Large$\m@th\boxplus$}}}
  {\vcenter{\hbox{\large$\m@th\boxplus$}}}
  {\vcenter{\hbox{\normalsize$\m@th\boxplus$}}}
  {\vcenter{\hbox{\scriptsize$\m@th\boxplus$}}}
}\displaylimits}
\makeatother

\makeatletter
\renewcommand\footnoterule{%
  \kern -3pt
  \hrule width 1.6in height 0.4pt
  \kern 3pt
}
\makeatother

\usepackage{microtype}

\begin{document}
\addtocontents{toc}{\protect\setcounter{tocdepth}{-1}} 

\author{Constantin Cedillo Vayson de Pradenne}
\thanks{Equal contribution}
\email{ccedillo@caltech.edu}
\affiliation{Harvard Quantum Initiative, 60 Oxford St, Cambridge, MA 02138}
\affiliation{California Institute of Technology, Pasadena, CA, 91125}

\author{Ishaan Kannan}
\thanks{Equal contribution}
\email{ishaan_kannan@g.harvard.edu}
\affiliation{Harvard Quantum Initiative, 60 Oxford St, Cambridge, MA 02138}

\author{Harald Putterman}
\affiliation{Department of Physics, Harvard University, Cambridge, MA 02138 USA}

\author{Jordan Cotler}
\email{jcotler@fas.harvard.edu}
\affiliation{Harvard Quantum Initiative, 60 Oxford St, Cambridge, MA 02138}
\affiliation{Department of Physics, Harvard University, Cambridge, MA 02138 USA}

\title{\textls[-20]{Restrictions on non-Clifford fault tolerance and ruling out beyond-SQL quantum metrology}}

\begin{abstract}
Quantum metrology promises a quadratic speedup over the standard quantum limit (SQL), but signal-aligned noise is expected to preclude this advantage in realistic settings. A potential route around known no-go results is to encode the sensors in a quantum code where the physical signal acts transversally as a logical gate. Understanding restrictions on transversal non-Clifford gates is therefore central to both quantum metrology and fault-tolerant quantum computation. Here, we prove such restrictions and apply them to transversal sensing. For any stabilizer code of distance $d\ge 3$ supporting a transversal logical action in level $D$ of the Clifford hierarchy, every stabilizer generating set must contain a check of weight at least $2^D$. Moreover, any $r$-level concatenated realization satisfies $r\leq \lfloor \log_2 n/D\rfloor$, forcing $r=1$ and ruling out concatenation when applied to beyond-SQL metrology. We then show that transversal single-qubit rotations by a small angle $\theta$ can only induce a nontrivial logical action on an $n$-qubit code if its checks include irreducible stabilizers of weight $\Omega(1/(n|\theta|^2))$. Here, many single-qubit errors commute with every stabilizer or logical Pauli below this weight and are only detected by a high-weight check, so their syndromes cannot be fault-tolerantly reconstructed from low-weight normalizer measurements. Since beyond-SQL transversal sensing requires $|\theta| = o(n^{-1/2})$, the weight of checks required for syndrome extraction diverges with $n$. Finally, we prove a broader metrological no-go theorem that avoids the assumptions of the quantum Cram\'{e}r-Rao bound: constant-strength signal-aligned noise rules out any asymptotic advantage over the SQL in AC or DC sensing, even with biased estimators, nonstabilizer or approximate encodings, quantum memory, intermediate measurements, or adaptive control.
\end{abstract}

\maketitle

\begin{bibunit}
\setcounter{tocdepth}{-1} 

\section{Introduction}
\vspace{-2mm}

Quantum metrology has traditionally leveraged entangled probes to improve parameter-estimation precision from the standard quantum limit (SQL) to the quadratically enhanced Heisenberg limit (HL) \cite{Giovannetti_2004, Giovannetti_2006}. Noise poses a fundamental obstacle to this enhancement, since even weak local decoherence can destroy the correlations that support Heisenberg scaling, and quantum error correction (QEC) therefore provides a natural approach to preserving metrological advantage in future quantum sensors \cite{Kessler_2014}.

In canonical metrology, an unknown parameter couples to a known Pauli generator $\sigma$ through either a DC Hamiltonian $H = \omega \sigma$ or an AC Hamiltonian $H(t) = B\sin(\nu t+\phi)\sigma$. Under Markovian noise, QEC restores quadratic scaling of the quantum Fisher information (QFI) if and only if $\sigma$ lies outside the Lindblad span generated by the noise operators, a criterion known as \textit{Hamiltonian-not-in-Lindblad-span} (HNLS) \cite{Zhou_2018}. In widely studied settings such as phase sensing under dephasing, HNLS is violated, forcing the QFI to scale at most linearly with the number of sensors and ruling out an asymptotic quantum speedup within the quantum Cram\'er-Rao framework \cite{Helstrom1976QuantumDetection}. Because realistic noise often acts through the same generator support as the signal, HNLS appears to severely restrict the scope of metrological quantum advantage \cite{Layden_2018}.

This restriction, however, does not apply universally. The quantum Cram\'er--Rao bound (QCRB), which yields a sensing-time lower bound $T = \Omega({\rm QFI}^{-1/2})$, governs unbiased, local estimators which vary as a smooth function of the underlying parameter~\cite{Helstrom1967MinimumMeanSquaredError, BraunsteinCaves1994StatisticalDistance}. A QEC-protected protocol designed to distinguish a discrete set of parameter values need not define a locally unbiased estimator on a differentiable family of output states, and may therefore achieve beyond-SQL efficiency even when the signal generator lies within the Lindblad span. Whether error-corrected sensing can evade the HNLS no-go in the practical setting where signal and noise are \textit{aligned} (i.e., share Pauli generators) is therefore still open.

\begin{figure*}[t!]
    \centering
    \includegraphics[width=\linewidth]{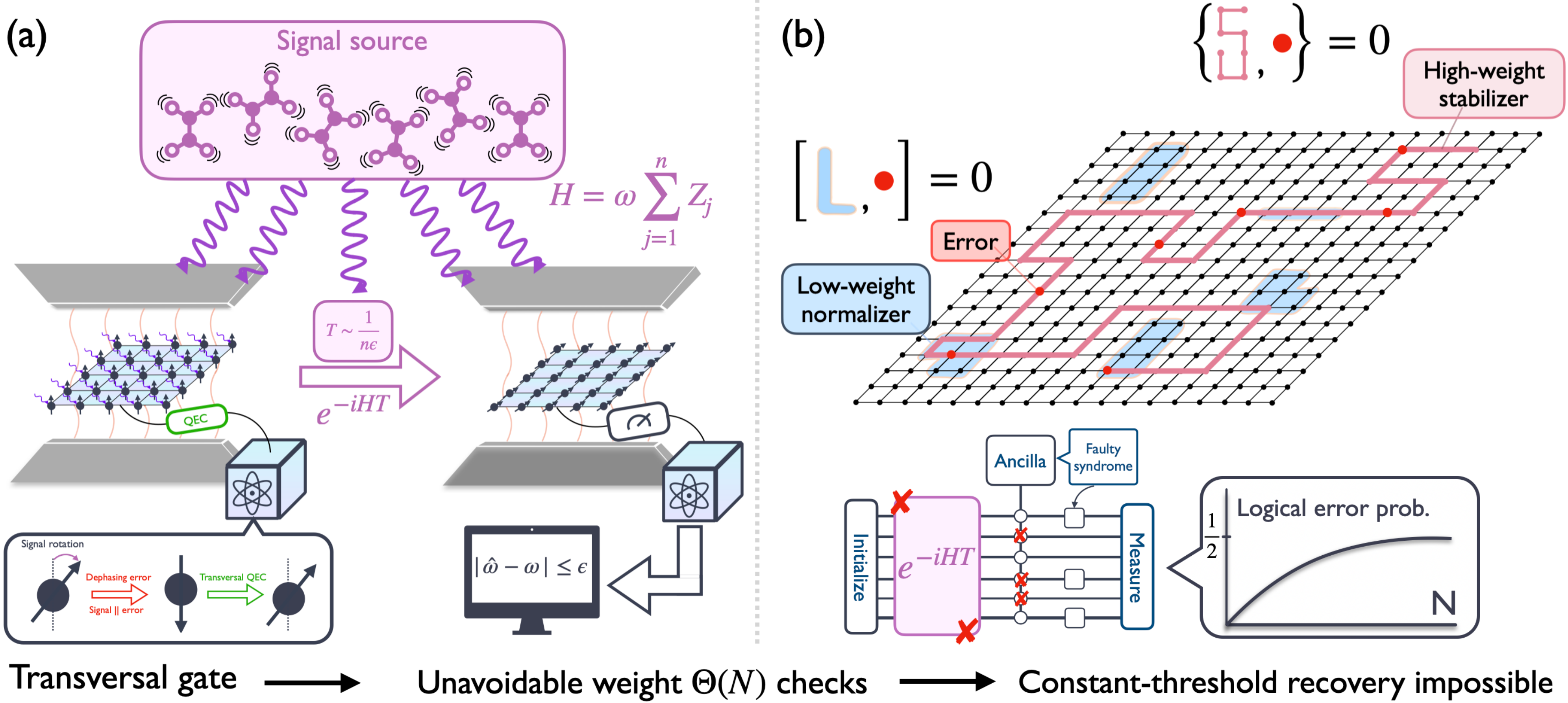}
    \caption{\textbf{Obstructions to transversal quantum sensing.} (a) \textit{Heisenberg-limited transversal sensing.} A spatially homogeneous signal acts on an array of $n$ aligned quantum sensors initialized in a code state, which evolves for a time of order $1/(n\epsilon)$. Near selected interrogation times, QEC removes local errors aligned with the signal generator while preserving the coherent global rotation as a discrete transversal logical gate; measuring the encoded state would then determine $\omega$ to precision $\epsilon$. (b) \textit{Fault-tolerance obstructions.} Any code family supporting transversal sensing in time $O(1/(n^{\alpha}\epsilon))$ for any $\alpha > 1/2$ must contain stabilizer checks of $\textnormal{poly}(n)$-growing weight, and many single-qubit errors can be detected only by such high-weight stabilizers, since every low-weight normalizer commutes with those errors. The resulting need to extract syndromes from increasingly nonlocal checks prevents asymptotically fault-tolerant error correction under local noise of any fixed strength.}
    \label{fig:1}
\end{figure*}

\textit{Transversal quantum sensing} provides a natural way to exploit this possibility [Fig.~\ref{fig:1}(a)]. In this picture, motivated by fault-tolerant quantum computation with transversal gates \cite{Gottesman_1998, Zhou_2025}, physical sensing qubits are prepared in the codespace of a quantum error-correcting code (QECC) such that the single-qubit rotations generated by the signal Hamiltonian $H = f(\omega)\sigma$, for a known function $f$, implement a nontrivial logical gate transversally. Although the Eastin--Knill theorem \cite{Eastin_2009} rules out QECCs supporting a continuous family of transversal logical gates, sensing requires only a discrete family of transversal rotations whose physical angle decreases with the code size; the resulting protocol distinguishes a discrete set of parameter values and need not yield an unbiased estimator or respond continuously to $\omega$, placing it outside the estimation setting governed by the HNLS no-go. The remarkable progress in high-rate QECCs admitting transversal non-Clifford gates \cite{Zhu_2025, golowich2024qldpc, nguyen2024goodbinaryquantumcodes, golowich2024agqc, he2025asymptoticallygoodquantumcodes} has thus far suggested that codes useful for transversal sensing may eventually be discovered.

Code families with the appropriate parameters would yield the first asymptotic beyond-SQL sensing protocols under realistic noise, while also producing QECCs with rich transversal gate sets useful for universal fault-tolerant quantum computation. Conversely, quantitative restrictions on such codes would elucidate how the non-Clifford complexity of transversal gates constrains QECC design and would exclude robust beyond-SQL metrology far more generally than the QCRB alone.

In this work, we derive such restrictions and apply them to any code family capable of transversal sensing in time $T = O(1/(\epsilon \,n^\alpha))$ with $\alpha > 1/2$. We show that the required transversal action necessitates genuinely nonlocal stabilizer structure, including high-weight checks whose syndrome information cannot be reconstructed from lower-weight normalizer measurements [Fig.~\ref{fig:1}(b)], and that concatenation cannot circumvent this obstruction. Since $\alpha = 1/2$ corresponds to the SQL, these constraints obstruct not only Heisenberg scaling but every asymptotic quantum speedup over the SQL within the transversal-sensing framework. We then prove a broader no-go theorem excluding beyond-SQL metrology whenever signal and noise align, closing the remaining route to noise-robust asymptotic advantage.

\vspace{-2em}
\vspace{2mm}
\section{Results}

\vspace{-2mm}
\subsection{Definitions and notation}
\vspace{-2mm}

\textit{Clifford hierarchy and stabilizer QEC~\cite{gottesman1997stabilizercodesquantumerror}} --- Rotations by dyadic angles about Pauli axes lie in the \textit{Clifford hierarchy}. For each $m \in \mathbb{Z}_{>0}$, let $\mathcal P_m$ denote the Pauli group on $m$ qubits.

\begin{definition}[Clifford hierarchy]
The Clifford hierarchy is defined recursively by $\mathcal C_1^{(m)} := \mathcal P_m$ and, for $D \ge 2$,
\begin{equation}
\mathcal C_D^{(m)} := \left\{U \in \mathrm U(2^m) : U P U^\dagger \in \mathcal C_{D-1}^{(m)} \text{ for all } P \in \mathcal P_m\right\}.
\end{equation}
\end{definition}
\noindent The second level, $\mathcal C_2^{(m)}$, is the Clifford group, while every gate in $\mathcal C_D^{(m)} \setminus \mathcal C_{D-1}^{(m)}$ with $D \ge 3$ is non-Clifford. We suppress the superscript whenever the number of qubits is clear and say that a gate is \emph{genuinely level $D$} if it belongs to $\mathcal C_D \setminus \mathcal C_{D-1}$.

Let $\mathcal C$ be an $[[n,k,d]]$ stabilizer code with logical Hilbert space $\mathcal L \cong (\mathbb C^2)^{\otimes k}$ and stabilizer group $S$, and let $\mathcal N(S)$ denote the normalizer of $S$ in the $n$-qubit Pauli group. Fix a basis of logical Pauli representatives $\{\widetilde X_a,\widetilde Z_a\}_{a=1}^k \subset \mathcal N(S)$ and an encoding isometry $V:\mathcal L\to(\mathbb C^2)^{\otimes n}$ satisfying $\widetilde X_a V = V X_a$ and $\widetilde Z_a V = V Z_a$ for every $a\in[k]$, and let $\Pi:=VV^\dagger$ denote the codespace projector. A physical unitary $U$ preserves the code if $U\Pi U^\dagger = \Pi$, in which case it induces the logical unitary $\overline U := V^\dagger UV$. For any Pauli operator $Q$, we write $\mathrm{wt}(Q)$ for the number of qubits on which $Q$ acts nontrivially. A generating set for $S$ is a subset $\mathcal G \subseteq S$ such that $\langle \mathcal G\rangle = S$; its elements are called stabilizer checks. For each integer $L \ge 1$, let $\mathcal N_{\le L}(S) := \left\langle Q \in \mathcal N(S) : \mathrm{wt}(Q) \le L \right\rangle$ denote the subgroup of $\mathcal N(S)$ generated by normalizer elements of weight at most $L$.

\begin{definition}[Minimax stabilizer and normalizer stabilizer weights]
Let $S$ be a stabilizer group. Its minimax stabilizer weight is
\begin{equation}
\lambda_S := \min_{\mathcal G \subseteq S\,:\,\langle \mathcal G\rangle = S}\ \max_{Q \in \mathcal G} \mathrm{wt}(Q),
\end{equation}
and its minimax normalizer weight is
\begin{equation}
\lambda_S^{\mathcal N} := \min\left\{L \ge 1 : S \subseteq \mathcal N_{\le L}(S)\right\}.
\end{equation}
\end{definition}
\noindent Thus $\lambda_S$ is the smallest achievable value of the largest check weight among all generating sets of $S$, quantifying the irreducible nonlocality of any stabilizer presentation of the code, while $\lambda_S^{\mathcal N}$ is the smallest $L$ for which the commutation syndrome of every Pauli error with every stabilizer can be reconstructed from its commutation relations with normalizer elements (i.e., logical operators and stabilizers) of weight at most $L$. Since every stabilizer is a normalizer element, $\lambda_S^{\mathcal N} \le \lambda_S$.

\textit{Distance measures} --- We write $\|\cdot\|_1$ for the trace norm, $\|\cdot \|_{\rm op}$ for the operator norm, $d_B(\rho,\sigma) := \arccos\|\sqrt{\rho}\sqrt{\sigma}\|_1$ for the Bures angle between quantum states, and $d_{\mathrm{ch}}$ for the corresponding channel distance \cite{Yuan_Fung_2017}.

\textit{Quantum metrology} --- We consider estimation of a single parameter $\omega$ to within absolute error $\epsilon$ (with high probability) using $n$ sensors. Heisenberg-limited protocols use total signal exposure time $O(1/(\epsilon\,n))$, whereas protocols using larger times $O(1/(\epsilon\,\sqrt{n}))$ or $O(1/(\epsilon^2 n))$ are standard-quantum-limited; both scalings are achievable by entanglement-free protocols in different regimes. Any protocol that improves either the $\sqrt{n}$ factor or the $\epsilon^2$ factor achieves beyond-classical performance, and our results encompass all such approaches.

\vspace{-2mm}
\subsection{Restrictions on non-Clifford fault tolerance}
Because smaller-angle rotations lie increasingly high in the Clifford hierarchy, we first derive quantitative restrictions on transversal implementations of non-Clifford logical gates: implementing a genuinely level-$D$ logical gate requires stabilizer checks of weight at least $2^D$, so the irreducible nonlocality of the stabilizer group grows exponentially with the level of the induced logical gate, making fault-tolerant implementation increasingly difficult. Our results complement the canonical no-go theorems for non-Clifford gates of Refs.~\cite{Bravyi_2013, Pastawski_2015, Baspin_2022}.

\begin{theorem}\label{thm:clifford-hierarchy-wb}
Let $D\ge 1$. Let $U = \bigotimes_{j=1}^n U_j$, with $U_j \in \mathrm U(2)$, preserve a stabilizer code $\mathcal C$ of distance $d \ge 3$. If the induced logical unitary $\overline U$ is genuinely level $D$, then $\lambda_S \ge 2^D$.
\end{theorem}

\noindent The theorem imposes no restriction on the individual physical factors beyond single-qubit unitarity; only the induced logical action is required to lie in $\mathcal C_D \setminus \mathcal C_{D-1}$. The bound is also tight in its dependence on $D$, since there exist stabilizer codes admitting transversal implementations of genuinely level-$D$ logical gates with minimax stabilizer weight $\lambda_S = 2^D$ \cite{landahl2013complexinstructionsetcomputing, bombin2015gaugecolorcodesoptimal}.

An established strategy for achieving fault tolerance despite the presence of high-weight stabilizer measurements is concatenation \cite{aharonov1999faulttolerantquantumcomputationconstant}. However, we next show that concatenated codes with significant recursive depth cannot implement large-$D$ transversal non-Clifford gates, excluding concatenation as a workaround to Theorem~\ref{thm:clifford-hierarchy-wb}.

\begin{theorem}\label{thm:concatenated_no_go}
Let $\mathcal{C} = \mathcal{C}_1\circ \mathcal{C}_2\circ \cdots \circ \mathcal{C}_r$ be an $r$-level concatenated code where every constituent $\mathcal{C}_j$ is an $[[n_j, k_j, d_j]]$ stabilizer code with $d_j \geq 3$, and $k_j = 1$ for $j\geq 2$. Let $n = \prod_j n_j$ denote the total number of physical qubits. If any $U = \bigotimes_{j=1}^n U_j$ with $U_j\in \mathrm U(2)$ induces a logical unitary $\overline{U}$ that is genuinely level $D$, then $r \leq \lfloor \log_2 n / D\rfloor$.
\end{theorem}

\noindent Combined with the fact that dyadic rotations whose induced logical action is genuinely level $D$ require $\theta = \Theta(n^{-\alpha}) = \Theta(2^{-D})$, i.e., $D = \alpha \log_2(n)+O(1)$, Theorem~\ref{thm:concatenated_no_go} implies that, for beyond-SQL scaling ($\alpha>1/2$), concatenated code families cannot support transversal sensing unless $r\le 1$, for all sufficiently large $n$. Beyond their relevance to QECC design, Theorems~\ref{thm:clifford-hierarchy-wb} and \ref{thm:concatenated_no_go} illustrate that if a transversal sensing protocol achieving time $O(1/(\epsilon\, n^\alpha))$ relies on genuine non-Clifford logical actions, it must measure stabilizers whose weights grow as $\sim n^\alpha$ without additional concatenated structure, impeding fault-tolerant implementation. We therefore turn to constraints on codes for which the logical action induced by small-angle single-qubit rotations can be entirely arbitrary.

\begin{theorem}\label{thm:discrete-eastin-knill}
Let $\mathcal C$ be an exact $n$-qubit stabilizer code of distance $d \ge 3$ with stabilizer group $S$. For nonidentity single-qubit Pauli operators $\sigma_j$ acting on qubit $j$, define $G := \sum_{j=1}^n \sigma_j$ and $U_\theta := e^{-i\theta G}$. Suppose that $U_\theta$ preserves $\mathcal C$ and induces a logical unitary $\overline U_\theta$ satisfying $\inf_{z \in \mathbb C}\left\|\overline U_\theta - z I_{\mathcal{L}}\right\|_{\rm op} \ge s$ for some $s > 0$. Then
\begin{equation}
\lambda_S \ge \lambda_S^{\mathcal N} \ge \frac{s^2}{n|\theta|^2}\,.
\end{equation}
Consequently, if $s = \Omega(1)$ and $|\theta| = O(n^{-\alpha})$, then $\lambda_S = \Omega\!\left(n^{2\alpha-1}\right)$ and $\lambda_S^{\mathcal N} = \Omega\!\left(n^{2\alpha-1}\right)$.
\end{theorem}
\noindent A transversal product of small single-qubit rotations can therefore induce a nontrivial logical transformation only if the stabilizer structure of the code becomes increasingly nonlocal, independent of the particular logical gate implemented. This provides an algebraic manifestation of the nonlocal correlations required for coherent signal accumulation: a substantial logical rotation generated by small local angles forces significant entanglement in the codespace, which in turn induces irreducible nonlocal structure in the stabilizer group. This argument is made formal in Appendix \ref{app:thms_2_3}. Applied to transversal sensing protocols with beyond-SQL scaling, which operate in time $T = O(1/(\epsilon\, n^\alpha))$ for some $\alpha > 1/2$, Theorem~\ref{thm:discrete-eastin-knill} forces both $\lambda_S$ and $\lambda_S^{\mathcal N}$ to grow polynomially with the code size.

We next establish that these polynomially large stabilizers are truly essential to performing QEC in sensing-compatible codes, and that no syndrome measurement strategy can avoid faulty high-weight measurements by only measuring high-weight stabilizers infrequently. This formally excludes many fault-tolerant syndrome-extraction strategies which avoid high-weight measurements by inferring their syndromes from low-weight logical operators \cite{Aliferis_2007, Bravyi_2011}.

\begin{theorem}[Transversal syndrome irreducibility]\label{thm:lowweight-syndr-incompleteness}
Under the hypotheses of Theorem~\ref{thm:discrete-eastin-knill}, for every integer $L\ge1$  such that $L <  \frac{s^2}{n|\theta|^2}$, there exist a subset $J \subseteq [n]$ with $|J| > L$ and a stabilizer $P \in S \setminus \mathcal N_{\le L}(S)$ such that
\begin{enumerate}
\item $[\sigma_j,Q] = 0$ for every $j \in J$ and every $Q \in \mathcal N_{\le L}(S)$;
\item $\{\sigma_j,P\} = 0$ for every $j \in J$.
\end{enumerate}
\end{theorem}

\noindent When small-angle rotations are used to generate transversal logical gates, the resulting code therefore has high-weight stabilizer generators whose measurement is unavoidable for correcting basic single-qubit errors: in the small-angle regime $|\theta| = o(n^{-1/2})$ relevant to transversal sensing, there are $\textnormal{poly}(n)$-many single-qubit errors that can only be detected by measurement of a stabilizer whose weight also grows polynomially in $n$. The same constraint is relevant to non-Clifford fault tolerance, because avoiding increasingly high-weight stabilizer measurements requires transversal non-Clifford gates to be implemented using single-qubit rotations whose angles are sufficiently large that the bound on stabilizer weight does not grow with the code size.

As an intuitive illustration of how Theorems~\ref{thm:discrete-eastin-knill} and \ref{thm:lowweight-syndr-incompleteness} impede practical fault-tolerant realizations of transversal sensing, consider an idealized Shor-style syndrome extraction \cite{shor1997faulttolerantquantumcomputation} in which a perfectly prepared $w$-qubit GHZ ancilla is coupled to $w$ data qubits supporting an irreducible weight-$w$ check, with faults occurring independently with probability $p$ at each data--ancilla interaction and nowhere else. The measured syndrome is flipped when an odd number of faults occur, so the syndrome-extraction error probability is $\frac{1}{2}-(1-2p)^w/2$; with only a constant number of gadget repetitions, keeping this below any constant less than $1/2$ requires $p = O(1/w)$, so the maximum tolerable fault probability per data--ancilla interaction vanishes as the stabilizer weight grows. Moreover, because the stabilizer is irreducible in the sense of Theorem~\ref{thm:lowweight-syndr-incompleteness}, any strategy for replacing the high-weight check by a collection of lower-weight measurements is unavailable.

\vspace{-1em}
\subsection{No-go theorem for beyond-SQL metrology}\label{sec:transverse-sens-no-go}

The previous section established stringent constraints on stabilizer codes with transversal non-Clifford gates, but more general approaches to beyond-SQL metrology may remain via exotic protected sensing schemes using e.g. nonstabilizer codes, approximate covariant QEC \cite{Liu_2023, elovenkova2026covariantapproximatequantumcodes}, distillation \cite{Bombin_2006}, or weight reduction \cite{hastings2023quantumweightreduction}. We now formally rule out all approaches to asymptotic beyond-SQL metrology in the presence of signal-aligned noise, beginning with a weaker but practically illustrative consequence of Theorem~\ref{thm:discrete-eastin-knill}.

Any sensor that estimates a parameter $\omega$ to precision $\epsilon$ with constant success probability can be applied to distinguish the two parameters $\omega_{\pm} = \omega_0 \pm \epsilon$. Prepare the same encoded probe state $\rho_0$ in each case and evolve it for time $T$ under the signal Hamiltonians $H_{\pm} = \omega_{\pm}G$. After removing the common evolution generated by $\omega_0G$, the two hypotheses differ by the transversal unitary $W_T := e^{-2i\epsilon T G}$. Letting $\overline W_T$ denote its induced logical action and $\rho_{\pm}$ the two final probe states, Helstrom's bound \cite{Helstrom1976QuantumDetection} gives that the optimal distinguishing probability is bounded by 
\begin{equation}
\frac{1}{2} + \frac{1}{4}\|\rho_+ - \rho_-\|_1 \le \frac{1}{2}\left(1 + \inf_{z \in \mathbb C}\left\|\overline W_T - z I_{\mathcal{L}}\right\|_{\rm op}\right) \ .
\end{equation}
Consequently, distinguishing the two frequencies with success probability at least $1/2 + p$ requires $\inf_{z \in \mathbb C}\|\overline W_T - z I_{\mathcal{L}}\|_{\rm op} \ge 2p$. For constant $p>0$, Theorem~\ref{thm:discrete-eastin-knill}, applied with rotation angle $2\epsilon T$, implies that any transversal DC sensing protocol must satisfy
\begin{equation}\label{eq:transversal_cw_sql}
T = \Omega\!\left(\frac{1}{\epsilon\sqrt{n\lambda_S^{\mathcal N}}}\right),
\end{equation}
because distinguishability is a prerequisite for metrology. Conversely, a transversal sensing protocol with interrogation time $T = O(1/(\epsilon n^\alpha))$ must satisfy $\lambda_S \ge \lambda_S^{\mathcal N} = \Omega(n^{2\alpha-1})$: for every $\alpha > 1/2$, the required stabilizer-check weight diverges with the code size, so no family of stabilizer codes with bounded-weight checks can support beyond-SQL transversal sensing.

While this bound rules out many practical encoding schemes, it does not formally exclude fault-tolerant constructions with growing-weight checks, nonstabilizer codes, or approximate error correction. We next prove a sensing-specific no-go theorem that excludes all three possibilities and applies beyond transversal sensing to arbitrary metrological protocols.

In this model, signal-aligned dephasing acts during each sensing interval, while every transition from sensing to quantum processing subjects each sensor to independent erasure, representing a fixed local noise floor at the processor--sensor interface. Such control-associated noise is essential in any metrological model permitting fast local processing of the sensors, since operations at the processor--sensor interface are unprotected and cannot be executed arbitrarily quickly; Refs.~\cite{CotlerGongKannan2026NoisyQuantumLearning, kannan2026exponentialspeedupsfaulttolerantprocessing, romanov2026learningarbitrarylindbladiansquantum} and Appendix~\ref{app:no-go-for-sensing} further discuss why this interface noise is physically unavoidable. Any fixed-strength local interface noise produces the same asymptotic scaling in $n$ and $\epsilon$, and our result already precludes beyond-SQL sensing without control-associated errors, but including them yields a more operationally complete bound. Subject to these noise processes, we allow the protocol to use arbitrary quantum memories and ancillas, intermediate controls and measurements, adaptive feedback, and any number of sensing rounds; the model therefore encompasses arbitrary QEC and error-mitigation procedures, since quantum and classical processing between signal queries is unrestricted.

\begin{theorem}[No-go theorem for beyond-SQL metrology, simplified]\label{thm:general-tranverse-sensing-no-go_main}
Consider any multi-round sensing protocol in which round $r$ evolves for time $t_r$ under the signal Hamiltonian $H_\omega = \omega G$ while subject to independent local dephasing of strength $\gamma > 0$ aligned with the signal generator. Between rounds the protocol may use arbitrary quantum memory, control, measurements, and adaptivity. At every interface between signal interrogation and quantum processing, suppose that each sensing qubit independently undergoes a flagged erasure with probability $p_{\mathrm{e}}\in (0, 1)$. If $\rho_\omega^{\mathrm{out}}$ denotes the final state produced by the protocol for signal Hamiltonian $H_\omega$, then for any $\omega_0,\omega_1$ with $\epsilon = |\omega_0-\omega_1|$, we have
\begin{equation}
d_B\!\left(\rho_{\omega_0}^{\mathrm{out}},\rho_{\omega_1}^{\mathrm{out}}\right) \leq \epsilon \min\left\{ nT,\, T\sqrt{\frac{n(1-p_{\mathrm{e}})}{p_{\mathrm{e}}}},\, \frac{1}{2}\sqrt{\frac{nT}{\gamma}} \right\},
\end{equation}
where $T = \sum_r t_r$ is the total interrogation time.
\end{theorem}
\noindent Combining Theorem~\ref{thm:general-tranverse-sensing-no-go_main} with the Helstrom distinguishability requirement yields the lower bound
\begin{equation}\label{eq:T-lower-bound}
T = \Omega\!\left(\max\left\{\frac{1}{\epsilon\,n},\,\sqrt{\frac{p_{\mathrm{e}}}{1-p_{\mathrm{e}}}}\,\frac{1}{\epsilon\sqrt{n}},\,\frac{\gamma}{\epsilon^2 n}\right\}\right)
\end{equation}
on the total interrogation time of an arbitrary sensing protocol, yielding a direct tradeoff between signal-aligned noise and sensing efficiency. The first term is the Heisenberg limit and can dominate only in a nonasymptotic regime where both noise rates are sufficiently small relative to $n$ and $1/\epsilon$. The second reproduces the coherent-evolution $1/(\epsilon\sqrt{n})$ SQL scaling obtained from Eq.~\eqref{eq:transversal_cw_sql} for bounded $\lambda_S^{\mathcal N}$, while the third is the familiar $1/(n\epsilon^2)$ SQL for metrology under Markovian dephasing with repeated interrogation \cite{Demkowicz_Dobrza_ski_2012}. For error rates fixed as $n$ and $\epsilon$ vary, the coherent-evolution term dominates the Heisenberg term beyond small $n$, and the Markovian term dominates it when $\epsilon$ is small. Both asymptotic scalings are achievable by entanglement-free sensing strategies, so Theorem~\ref{thm:general-tranverse-sensing-no-go_main} formally rules out any asymptotic quantum metrological speedup under signal-aligned noise; the result extends straightforwardly to any setting with a constant noise floor in the signal direction, whether Markovian or interstitial, and thereby circumscribes most practical metrological settings. 

Unlike the HNLS condition of Ref.~\cite{Zhou_2018}, which excludes beyond-SQL scaling only for unbiased estimators governed by the QCRB, Theorem~\ref{thm:general-tranverse-sensing-no-go_main} applies directly to the output-state distinguishability of arbitrary single-parameter sensing protocols---with or without the use of QEC---and therefore closes the gap left by QCRB-based no-go theorems: when the signal and local noise act through the same generators, no protocol can achieve an asymptotic quantum advantage. Appendix~\ref{app:no-go-for-sensing} establishes the corresponding formal statements for both DC and AC sensing.

A special case of the distinguishability bound exposes the mechanism that prevents transversal sensing from correcting signal-aligned noise. Let $\mathcal E$ be an encoding channel into $n$ physical qubits, and suppose that $U_\theta = e^{-i\theta G}$, with $G = \sum_{j=1}^n \sigma_j$, implements a logical unitary $\overline U_\theta$ satisfying $\inf_{z \in \mathbb C}\|\overline U_\theta - z I_{\mathcal{L}}\|_{\mathrm{op}} \ge s > 0$. For an interrogation time $T$, let $\mathcal N_{\gamma,T}^{\mathrm{deph}} = \bigotimes_{j=1}^n \mathcal N_{\gamma,T}^{(j),\mathrm{deph}}$ denote independent Markovian dephasing aligned with the signal generator, and assume that all quantum processing is noiseless. Define the optimal logical correction error
\begin{equation}
\delta_n(\gamma,T) := \inf_{\mathcal R} d_{\mathrm{ch}}\!\left(\mathcal R \circ \mathcal N_{\gamma,T}^{\mathrm{deph}} \circ \mathcal E,\mathrm{id}_L\right),
\end{equation}
where the infimum is over all recovery channels $\mathcal R$. As shown in Appendix~\ref{app:no-go-for-sensing},
\begin{equation}
\delta_n(\gamma,T) \ge \frac{s}{4} - \frac{|\theta|}{4}\sqrt{\frac{n}{\gamma T}}\,.
\label{eq:transversal-dephasing-bound}
\end{equation}
Now suppose that $s = \Omega(1)$, $\gamma = \Theta(1)$, and $|\theta| = T\epsilon = O(n^{-\alpha})$ for some $\alpha > 1/2$, as required for beyond-SQL transversal sensing. The second term in Eq.~\eqref{eq:transversal-dephasing-bound} then satisfies $|\theta|\sqrt{n/(\gamma T)} = O(\sqrt{\epsilon \,n^{1-\alpha}})$ and vanishes in the high-precision regime $\epsilon = o(n^{\alpha-1})$, so
\begin{equation}
    \liminf_{n\rightarrow\infty}\delta_n(\gamma,T) \ge s/4
\end{equation} 
and the optimal logical correction error remains bounded away from zero. Outside this high-precision regime, the dephasing contribution to Theorem~\ref{thm:general-tranverse-sensing-no-go_main} requires $T = \Omega(\gamma/(\epsilon^2 n))$, which already has SQL scaling. No family of transversal sensing codes can therefore achieve beyond-SQL scaling while correcting signal-aligned noise with asymptotically vanishing logical error, even when nonstabilizer codes and approximate error correction are allowed. The same distinguishability bound implies that covariant codes supporting continuous logical $\mathrm U(1)$ rotations, such as those proposed in Refs.~\cite{Lin_2025,huang2026robustphasecontinuoustransversal}, cannot yield beyond-SQL metrology under signal-aligned noise whose strength remains nonzero asymptotically. 

\vspace{-2mm}
\section{Discussion}

We have shown that transversal non-Clifford gates, as well as logical gates generated by sufficiently small single-qubit physical rotations, require irreducible stabilizer structure of growing weight, making fault-tolerant syndrome extraction increasingly demanding with increasing Clifford-hierarchy level and decreasing physical rotation angle. We also derived a no-go theorem excluding asymptotic beyond-SQL metrology under local noise aligned with the signal generator, without assuming unbiased estimation or the applicability of the quantum Cram\'er-Rao bound.

At finite code sizes, our bounds quantify the stabilizer-check weights required by a given transversal logical action. They therefore provide quantitative design constraints complementary to the Eastin--Knill and Bravyi-K\"onig theorems and answer an open question posed in Ref.~\cite{chakraborty2026nogotheoremfaulttolerant}. It remains to identify practical code families that approach these bounds, to determine whether analogous nonlocality constraints hold for subsystem, nonstabilizer, and approximate codes admitting transversal non-Clifford gates, and to determine whether an irreducibility statement like Theorem~\ref{thm:lowweight-syndr-incompleteness} applies to the general setting of Theorem~\ref{thm:clifford-hierarchy-wb}.

Our metrological bounds have closed the gaps left by earlier no-go theorems for beyond-SQL error-corrected sensing \cite{Huelga_1997, Escher_2011, Demkowicz_Dobrza_ski_2012, Kessler_2014, Zhou_2018}, ruling out any asymptotic quantum speedup under nonvanishing signal-aligned noise. This shifts the emphasis in canonical AC and DC phase sensing toward useful finite-size advantages that nevertheless remain. In many cases entanglement can produce nearly quadratic improvements at experimentally relevant finite sensor numbers before noise removes the asymptotic gain \cite{Nichols_2016, Zhou_2020, Kielinski_2024, CotlerGongKannan2026NoisyQuantumLearning}, and transversal sensing may similarly protect useful encoded states for particular combinations of code size, encoding rate, physical noise, interrogation time, and target precision. Determining these crossover regimes, together with whether their preasymptotic gains justify the required control overhead, is therefore a natural practical objective.

More broadly, quantum advantage may persist in sensing tasks outside canonical single-parameter phase estimation. Realistic experiments seldom provide a perfectly known generator, a freely chosen interrogation time, and a single parameter whose local estimation error is the sole objective \cite{Tsang_2011, Bylander2011NoiseSpectroscopy, Macieszczak_2014}; the quantity of interest is often a nonlinear or transformed feature of a structured multiparameter signal rather than the parameters themselves \cite{cotler2026quantumadvantagesensingproperties}, and many established sensing platforms use bosonic modes, whose signal structure and control constraints differ substantially from those of qubit arrays \cite{Clerk_2010, Backes_2021, eickbusch2022fast}. A sensing framework centered on global inference, structured-signal learning, measurement complexity, or constrained access models may support superpolynomial and potentially noise-robust advantages through mechanisms distinct from Heisenberg scaling \cite{caves_1981, lloyd2008quantumillumination, Tan_2008, Tsang_2011, Aharonov_2022, cotler2026quantumadvantagesensingproperties}. These directions motivate a theory of quantum sensing organized around broader inference objectives, finite resources, and realistic experimental architectures.

\vspace{-1em}
\section*{Acknowledgements}
We thank Pablo Bonilla, Andrei Diaconu, Jin Ming Koh, Rohan Mehta, Nikita Romanov, Mikhail Lukin, Mincheol Park, and Daniel Tan for helpful discussions. CCVP is supported in part by a Caltech SURF Fellowship. JC is supported by a fellowship from the Alfred P.~Sloan Foundation. IK is supported in part by the Nobile Research Initiative. HP was supported by the Department of Defense through the National Defense Science and Engineering Graduate (NDSEG) Fellowship Program.

\putbib
\end{bibunit}
\clearpage
\onecolumngrid

\appendix

\makeatletter
\@removefromreset{equation}{section}
\@removefromreset{figure}{section}
\@removefromreset{table}{section}
\makeatother

\renewcommand{\thesection}{S\arabic{section}}
\renewcommand{\theequation}{S\arabic{equation}}
\renewcommand{\thefigure}{S\arabic{figure}}
\renewcommand{\thetable}{S\arabic{table}}

\setcounter{equation}{0}
\setcounter{figure}{0}
\setcounter{table}{0}
\setcounter{tocdepth}{2}

\begin{bibunit}
\input{supplement}

\putbib
\end{bibunit}

\end{document}

%% file: supplement.tex
\onecolumngrid
\startcontents[appendix]

\allowdisplaybreaks

\vspace*{1em}
\begin{center}
  {\large\bfseries Supplemental Material}
\end{center}
\vspace{1em}

\begingroup
  \hypersetup{hidelinks}

  \begin{center}
    \textbf{Contents}
  \end{center}
  \vspace{0.5em}

\printcontents[appendix]{}{1}{}
\endgroup

\makeatletter
\@removefromreset{equation}{section}
\makeatother
\renewcommand{\theequation}{S\arabic{equation}}

\section{Proofs of Theorems \ref{thm:clifford-hierarchy-wb} and \ref{thm:concatenated_no_go}}
First we prove Theorem \ref{thm:clifford-hierarchy-wb}. For brevity, we refer to a logical gate in
$\mathcal C_D\setminus\mathcal C_{D-1}$ as a level-$D$ gate. The physical one-qubit factors $U_j$ are arbitrary elements of $U(2)$ and need not themselves lie in the Clifford hierarchy. Our first step is to show that, after local Clifford changes of basis, any transversal unitary preserving a stabilizer code is equivalent on the codespace to a transversal diagonal gate with dyadic phases. For $\alpha\in\mathbb R$, define
\begin{equation}
    Z(\alpha):=\mathrm{diag}(1,e^{i\pi\alpha}),\qquad T_E:=Z(2^{1-E}),\qquad D_q^{(E)}:=\bigotimes_{j=1}^nT_E^{q_j},\qquad q\in(\mathbb Z/2^E\mathbb Z)^n.
\end{equation}

\subsection{Preliminary results}
\noindent For $D=1$, we interpret a genuinely level-$1$ gate as a nontrivial logical Pauli, up to phase. If $D=1$ and $\lambda_S=1$, then $S$ is generated by weight-one Paulis. Since the code encodes at least one qubit, some physical qubit is unstabilized and supports a weight-one logical Pauli, contradicting $d\ge 3$. Thus $\lambda_S\ge 2$. If $D=2$ and $\lambda_S<4$, then $S$ is generated by checks of weight at most three, which forces $d\le 2$~\cite{wang2026checkweight}, a contradiction. We therefore assume $D\ge 3$ throughout the rest of the section.
\begin{definition}
\label{def:bell-pair-factors}

A stabilizer code $\mathcal C$ has a tensor factor $\mathcal C_B$ on a
subset $B$ of the physical qubits if, after a relabeling of qubits,
    $\mathcal C=\mathcal C_{B^c}\otimes\mathcal C_B.$
A \emph{Bell-pair factor} of a stabilizer code is a two-qubit $[[2,0,2]]$ stabilizer-code tensor factor.
A \emph{trivially encoded qubit} is a $[[1,1,1]]$ tensor factor. A code is Bell-pair free, respectively free of trivially encoded qubits, if it has no tensor
factor of the corresponding type.
\end{definition}

\noindent
This terminology is from \cite[Def.~11]{ZengCrossChuang2011}.
A Bell-pair factor is an isolated pair of physical qubits whose state is fixed, up to local Clifford changes of basis, to a single Bell state. It carries no logical information, so any transversal unitary preserving the full codespace can act on this one-dimensional factor only by a phase. By
contrast, a trivially encoded qubit is a bare physical qubit carrying one logical qubit with no redundancy, and therefore forces the code distance to be one.
\begin{proposition}
\label{prop:zcc-local-normal-form}
Let $\mathcal C$ be a stabilizer code that is free of Bell-pair factors and trivially encoded qubits. If $W=\bigotimes_{j=1}^nW_j$ preserves $\mathcal C$, then, for every $j$, there exist $\gamma_j,\alpha_j\in\mathbb R$, a one-qubit Clifford gate $L_j$, and a nonidentity one-qubit Pauli $P_j\in\{X,Y,Z\}$ such that
\begin{equation}
    W_j=e^{i\gamma_j}L_je^{i\alpha_jP_j}.
\end{equation}
The Clifford case is included by taking $\alpha_j=0$.
\end{proposition}

\begin{proof}
For coordinates contained in a minimal stabilizer support, this is proved in \cite[Theorem~1]{ZengCrossChuang2011}. For all remaining coordinates, it is proved in
\cite[Lemma~5]{ZengCrossChuang2011}.
\end{proof}

\begin{proposition}
\label{prop:stabilizer-support-normal-form}
Let $\mathcal C$ be an $[[n,k]]$ stabilizer code.
\begin{enumerate}[label=(\roman*)]
    \item There is an orthonormal basis
    $\{|\psi_\ell\rangle\}_{\ell=1}^{2^k}$ of $\mathcal C$ whose
    computational-basis supports are pairwise disjoint and such that
    \begin{equation}
        |\psi_\ell\rangle =c_\ell\sum_{x\in S_\ell}i^{a_\ell(x)}|x\rangle,
    \end{equation}
    where $S_\ell\subseteq\mathbb F_2^n$, $c_\ell>0$, and
    $a_\ell:S_\ell\rightarrow\mathbb Z/4\mathbb Z$.

    \item Conversely, if
    $\{|\phi_\ell\rangle\}_{\ell=1}^{2^k}$ is any orthonormal basis of
    $\mathcal C$ with pairwise disjoint computational-basis supports,
    then
    \begin{equation}
        |\phi_\ell\rangle =e^{i\beta_\ell}d_\ell
    \sum_{x\in R_\ell}i^{b_\ell(x)}|x\rangle,
    \end{equation}
    where $R_\ell$ is the computational-basis support of
    $|\phi_\ell\rangle$, $d_\ell>0$, and
    $b_\ell:R_\ell\rightarrow\mathbb Z/4\mathbb Z$.
\end{enumerate}
\end{proposition}

\begin{proof}
Part~(i) is proved in
\cite[Lemma~1]{AndersonJochymOConnor2016}, and part~(ii) is proved in
\cite[Corollary~2]{AndersonJochymOConnor2016}.
\end{proof}

\begin{proposition}[Decompression]
\label{prop:diagonal-decompression}
Let $\mathcal C_0,\mathcal C_1$ be stabilizer codes of distance at
least two, and let $p_1,\ldots,p_n$ be positive integers. Set
$N:=\sum_{j=1}^np_j$ and define $J|x_1,\ldots,x_n\rangle:=\bigotimes_{j=1}^n|x_j\rangle^{\otimes p_j}$. If we have $\left(\bigotimes_{j=1}^nZ(p_j\alpha)\right)\mathcal C_0 =\mathcal C_1$, then $\widetilde{\mathcal C}_a:=J\mathcal C_a$, for $a\in\{0,1\}$, are stabilizer codes of distance at least two and $Z(\alpha)^{\otimes N}\widetilde{\mathcal C}_0 =\widetilde{\mathcal C}_1$.
\end{proposition}

\begin{proof}
The stabilizer-code and distance statements follow by iterating the decompression construction of \cite[Lemma~2 and Eqs.~(20)--(23)]{AndersonJochymOConnor2016}. The relation between the two encoded maps follows directly from $Z(\alpha)^{\otimes N}J=J\left(\bigotimes_{j=1}^nZ(p_j\alpha)\right)$.
\end{proof}

\begin{proposition}
\label{prop:strongly-transversal-dyadic}
Let $\mathcal C_0,\mathcal C_1$ be $N$-qubit stabilizer codes of distance at least two, each encoding the same positive number of logical qubits. If $Z(\alpha)^{\otimes N}\mathcal C_0=\mathcal C_1,$ then there exist integers $a$ and $r\ge0$ such that $\alpha\equiv\frac{a}{2^r}\pmod 2$.
\end{proposition}

\begin{proof}
This is proved in \cite[Proposition~4]{AndersonJochymOConnor2016}.
\end{proof}
The role of these preliminary results is as follows. Proposition~\ref{prop:zcc-local-normal-form} reduces the arbitrary one-qubit factors of a transversal unitary to Clifford gates followed by rotations about Pauli axes, once the irrelevant tensor factors
have been removed. Proposition~\ref{prop:stabilizer-support-normal-form} then writes stabilizer codewords as disjoint-support phase superpositions, allowing a diagonal transversal gate to be analyzed coefficientwise. Finally, once the diagonal angles have been shown to be rational, Propositions~\ref{prop:diagonal-decompression} and~\ref{prop:strongly-transversal-dyadic} convert them into a
strongly transversal rotation and force their denominators to be powers of two. Together, these results will be useful in Lemma \ref{lem:one-qubit-hierarchy-normal-form} to reduce an arbitrary transversal
automorphism to local Clifford changes of basis and a dyadic diagonal core.
\subsection{Dyadic normal form for transversal automorphisms}
\begin{lemma}[Transversal dyadic normal form]
\label{lem:one-qubit-hierarchy-normal-form}
Let $\mathcal C$ be an $n$-qubit stabilizer code of distance $d\ge 2$, with projector $\Pi$, and let $U = \bigotimes_{j=1}^n U_j$, for $U_j\in U(2)$, preserve $\mathcal C$. Then there exist tensor products of one-qubit Clifford gates $A = \bigotimes_{j=1}^n A_j$ and $B = \bigotimes_{j=1}^n B_j$, an integer $E\ge 2$, a vector $q\in(\mathbb Z/2^E\mathbb Z)^n$, and a phase $e^{i\varphi}$ such that
\begin{equation}
    U\Pi=e^{i\varphi}A D_q^{(E)}B\Pi\,.
\end{equation}
\end{lemma}

\begin{proof}
By recursively splitting off the Bell-pair factors from Definition~\ref{def:bell-pair-factors}, and then applying local Clifford gates on those factors and relabeling coordinates, one can write $\mathcal C=\mathcal C_{\mathrm c}\otimes \bigotimes_{\ell=1}^b\mathrm{span}\{|\Phi^+\rangle_\ell\}$. Indeed, every $[[2,0,2]]$ stabilizer code is local-Clifford equivalent to $\mathrm{span}\{|\Phi^+\rangle\}$. There are no trivially encoded qubits, since that would give a weight-one logical Pauli and hence distance one. Let $P_\ell:=|\Phi^+\rangle\!\langle\Phi^+|_\ell$ with $\Pi=\Pi_{\mathrm c}\otimes\bigotimes_{\ell=1}^bP_\ell$ and write $U=U_{\mathrm c}\otimes\bigotimes_{\ell=1}^bW_\ell$. Taking partial traces of $U\Pi U^\dagger=\Pi$ gives
\begin{equation}
    U_{\mathrm c}\Pi_{\mathrm c}U_{\mathrm c}^\dagger=\Pi_{\mathrm c}\,,\qquad W_\ell P_\ell W_\ell^\dagger=P_\ell.
\end{equation}
Since $P_\ell$ has rank one, there is a phase $e^{i\phi_\ell}$ such that $W_\ell P_\ell=e^{i\phi_\ell}P_\ell$. Thus every Bell-pair factor acts as a scalar on the codespace and can be replaced by the identity.
It remains to consider the Bell-pair-free core. If the core encodes no logical qubits, then its codespace is one-dimensional and the conclusion is immediate. Otherwise, Proposition~\ref{prop:zcc-local-normal-form} applies. After absorbing the phases of the individual tensor factors into one global phase, we may write $U_j=L_je^{i\alpha_jP_j},$ where $L_j$ is Clifford and $P_j\in\{X,Y,Z\}$. Choose a one-qubit Clifford gate $C_j$ satisfying $C_jP_jC_j^\dagger=Z$. Then \begin{equation}
    e^{i\alpha_jP_j}
= C_j^\dagger e^{i\alpha_jZ}C_j = e^{i\alpha_j}C_j^\dagger \mathrm{diag}(1,e^{-2i\alpha_j})C_j.
\end{equation}
Hence, up to a global phase, $U_{\mathrm c}=AD_{\boldsymbol\xi}B$, where $\xi_j=-2\alpha_j$ and $A,B$ are tensor products of one-qubit Clifford gates.  The only non-Clifford object left now is a diagonal transversal gate with potentially arbitrary real angles. Define $ \mathcal C_0 :=B\mathcal C_{\mathrm c},$ $\mathcal C_1:=A^\dagger\mathcal C_{\mathrm c},$ and $\Pi_0:=B\Pi_{\mathrm c}B^\dagger.$ One can see that $D_{\boldsymbol\xi}\mathcal C_0=\mathcal C_1.$  By Proposition~\ref{prop:stabilizer-support-normal-form}(i), choose an orthonormal basis $\{|\psi_\ell\rangle\}_{\ell=1}^{2^k}$ of
$\mathcal C_0$ with pairwise disjoint computational-basis supports. Thus, for each $\ell$, there are a set $S_\ell\subseteq\mathbb F_2^n$, a constant $c_\ell>0$, and a function  $a_\ell:S_\ell\rightarrow\mathbb Z/4\mathbb Z$ such that
\begin{equation}
    |\psi_\ell\rangle =c_\ell\sum_{x\in S_\ell}i^{a_\ell(x)}|x\rangle.
\end{equation}

Since $D_{\boldsymbol\xi}$ is diagonal and maps $\mathcal C_0$ onto $\mathcal C_1$, $|\phi_\ell\rangle:=D_{\boldsymbol\xi}|\psi_\ell\rangle$ form an
orthonormal basis of $\mathcal C_1$ and have the same pairwise disjoint supports $S_\ell$. By Proposition~\ref{prop:stabilizer-support-normal-form}(ii), there are phases $e^{i\beta_\ell}$ and functions $b_\ell:S_\ell\rightarrow\mathbb Z/4\mathbb Z$ such that
\begin{equation}
    |\phi_\ell\rangle =e^{i\beta_\ell}c_\ell
\sum_{x\in S_\ell}i^{b_\ell(x)}|x\rangle.
\end{equation}
The normalization constant is again $c_\ell$ because the support is $S_\ell$. Comparing the coefficients of $|x\rangle$ in $|\phi_\ell\rangle$, one gets 
\begin{equation}
    e^{i\boldsymbol\xi\cdot x}=e^{i\beta_\ell}i^{\,b_\ell(x)-a_\ell(x)}
\qquad \text{for all}\,\, x\in S_\ell.
\end{equation}
Consequently, for every $x,y\in S_\ell$,
$ \boldsymbol\xi\cdot(x-y)\in\frac{\pi}{2}\mathbb Z.$ Set $\boldsymbol t:=\boldsymbol\xi/\pi$. For each $\ell$, fix $x_\ell\in S_\ell$, and let $M$ be the integer matrix whose rows are the vectors
$x-x_\ell$, with $x\in S_\ell$ and $1\le\ell\le2^k$. From $\boldsymbol{\xi}\cdot (x-y)\in \frac{\pi}{2}\mathbb Z$, we get $M\boldsymbol t\in\frac{1}{2}\mathbb Z^m$. Since the linear system $M\boldsymbol u=M\boldsymbol t$ has integer coefficients and a rational right-hand side, it has a rational solution $\boldsymbol t_{\mathrm{rat}}\in\mathbb Q^n$. Define
$\boldsymbol t_{\mathrm{irr}}:=\boldsymbol t-\boldsymbol t_{\mathrm{rat}}$. Then
$M\boldsymbol t_{\mathrm{irr}}=0$, so
\begin{equation}\label{x-x-ell-tirr}
    \boldsymbol t_{\mathrm{irr}}\cdot x=\boldsymbol t_{\mathrm{irr}}\cdot x_\ell\qquad \text{for all}\,\, x\in S_\ell.
\end{equation}

Now, consider the one-parameter family
\begin{equation}
    D_{\mathrm{irr}}(s):=\bigotimes_{j=1}^n
\mathrm{diag}\!\left(1,e^{i\pi s(t_{\mathrm{irr}})_j}\right)=e^{isH_{\mathrm{irr}}},\qquad H_{\mathrm{irr}}:=\pi\sum_{j=1}^n(t_{\mathrm{irr}})_j|1\rangle\!\langle1|_j.
\end{equation}
This implies

\begin{align}
D_{\mathrm{irr}}(s)|\psi_\ell\rangle
&=c_\ell\sum_{x\in S_\ell}i^{a_\ell(x)}
\prod_{j=1}^n e^{i\pi s(t_{\mathrm{irr}})_j x_j}|x\rangle \\
&=c_\ell\sum_{x\in S_\ell}i^{a_\ell(x)}
e^{i\pi s\,t_{\mathrm{irr}}\cdot x}|x\rangle \\
&\overset{\eqref{x-x-ell-tirr}}{=}
e^{i\pi s\,t_{\mathrm{irr}}\cdot x_\ell}
c_\ell\sum_{x\in S_\ell}i^{a_\ell(x)}|x\rangle \\
&=e^{i\pi s\,t_{\mathrm{irr}}\cdot x_\ell}|\psi_\ell\rangle .
\end{align}
As such, $D_{\mathrm{irr}}(s)$ preserves $\mathcal C_0$ for every $s\in\mathbb R$. Hence,
$[D_{\mathrm{irr}},\Pi_0]=0$, and differentiating at $s=0$ gives $[H_{\mathrm{irr}},\Pi_0] = 0$. Since $\mathcal C_0$ has distance at least two, every one-qubit operator is detectable. Therefore, for each $j$, there is a scalar $\nu_j$ such that
\begin{equation}
\Pi_0|1\rangle\!\langle1|_j\Pi_0=\nu_j\Pi_0.
\end{equation}
It follows that $\Pi_0H_{\mathrm{irr}}\Pi_0=h\Pi_0$ for some $h\in\mathbb R$. Together with
$[H_{\mathrm{irr}},\Pi_0]=0$, this gives
$H_{\mathrm{irr}}\Pi_0=h\Pi_0$, and hence
\begin{equation}
    D_{\mathrm{irr}}(1)\Pi_0=e^{ih}\Pi_0.
\end{equation} Thus the irrational part acts as a global phase on $\mathcal C_0$.
This is similar to the two-code version of the rationality reduction of \cite[Appendix~A]{AndersonJochymOConnor2016}. Absorbing the phase, we can replace
$\boldsymbol\xi$ by $\pi\boldsymbol t_{\mathrm{rat}}$ and assume that every $\xi_j/\pi$ is rational.

Let $Q$ be a common denominator and choose
$p_j\in\{1,\ldots,2Q\}$ such that
$\frac{\xi_j}{\pi}\equiv\frac{p_j}{Q}\pmod 2.$  Thus
\begin{equation}
D_{\boldsymbol\xi}=\bigotimes_{j=1}^nZ\!\left(\frac{p_j}{Q}\right).
\end{equation}

Set $N:=\sum_jp_j$. Applying
Proposition~\ref{prop:diagonal-decompression} with $\alpha=1/Q$ gives stabilizer codes
$\widetilde{\mathcal C}_a:=J\mathcal C_a,$ for $ a\in\{0,1\},$ of distance at least two such that
\begin{equation}
    Z\!\left(\frac1Q\right)^{\otimes N}
\widetilde{\mathcal C}_0=\widetilde{\mathcal C}_1.
\end{equation}
Since the core encodes at least one logical qubit, Proposition~\ref{prop:strongly-transversal-dyadic} gives
\begin{equation}
    \frac{1}{Q}\equiv\frac{a}{2^r}\pmod 2
\end{equation} for some integers $a,r$. Since $1/Q$ is in lowest terms, $Q$ divides $2^r$, and hence $Q$ is a power of two. Therefore, for every $j$,
there exist integers $a_j,r_j$ such that

\begin{equation}
    \frac{\xi_j}{\pi} \equiv\frac{p_j}{Q}\equiv \frac{a_j}{2^{r_j}} \pmod 2.
\end{equation}

Choose $E\ge2$ such that $E-1\ge\max_jr_j$, and define $q_j:=a_j2^{E-1-r_j}\pmod{2^E}$. 
Then

\begin{equation}
    T_E^{q_j}= \mathrm{diag}\!\left(
1,e^{i\pi q_j/2^{E-1}} \right)=\mathrm{diag}\!\left(
1,e^{i\pi a_j/2^{r_j}}\right)=\mathrm{diag}(1,e^{i\xi_j}),
\end{equation}
Hence $D_{\boldsymbol\xi}\Pi_0= e^{i\phi}D_q^{(E)}\Pi_0.$
Consequently, \begin{equation}
    U_{\mathrm c}\Pi_{\mathrm c}=e^{i\varphi_{\mathrm c}}AD_q^{(E)}B\Pi_{\mathrm c}\,.
\end{equation}

Finally, reinserting the Bell-pair coordinates with identity Clifford factors and zero exponents, absorbing their scalar actions into the global phase, and undoing the preliminary local Clifford change of basis and coordinate relabeling gives
\begin{equation}
    U\Pi=e^{i\varphi}A D_q^{(E)}B\Pi\,.
\end{equation}
\end{proof}

\subsection{Descent of diagonal transversal gates}
 For a stabilizer group $S$, define \begin{equation}
     C_X(S):=\{x(P):P\in S\}, \qquad C_Z(S):=\{z\in\mathbb F_2^n:\pm Z(z)\in S\}.
 \end{equation}
For a function $F$ on an affine subspace of $\mathbb F_2^n$, define the finite difference 
\begin{equation}
    \Delta_\alpha F(x):=F(x+\alpha)-F(x).
\end{equation}

\begin{proposition}
\label{prop:reed-muller-duality}
For integers $0\le r<m$,
\begin{equation}
    \mathrm{RM}(r,m)^\perp =\mathrm{RM}(m-r-1,m),
\end{equation}
and $d_{\min}\!\left(\mathrm{RM}(r,m)\right)=2^{m-r}.$ Consequently, if $m\ge D$, every nonzero word in $\mathrm{RM}(D-1,m)^\perp$ has weight at least $2^D$.
\end{proposition}

\begin{proof}
This is proven in
\cite[Ch.~13]{MacWilliamsSloane}.
\end{proof}
\begin{lemma}
\label{lem:odd-phase-pairing}
    Let $H\in\mathbb F_2^{m\times n}$ have rows $r_1,\ldots,r_m$ of weight strictly less than $2^D$, and let $C_X=\mathrm{row}(H).$  Suppose $z\in\mathbb F_2^n$ satisfies
\begin{equation}
\sum_{j=1}^nz_j\prod_{s=1}^D(\alpha_s)_j=0 \pmod 2
\end{equation}
for every $\alpha_1,\ldots,\alpha_D\in C_X$. Then
\begin{equation}\label{z-decomposition}
    z=\sum_\mu c_\mu\,, 
\end{equation} where the vectors $c_\mu\in C_X^\perp$ have pairwise disjoint supports and weight at most two.
\end{lemma}
\begin{proof}
Let $h_j\in\mathbb F_2^m$ denote the $j$-th column of $H$. If $h_j=0$ and $z_j=1$, then $e_j\in C_X^\perp$, so $e_j$ is a weight-one term. It remains to consider a nonzero column pattern $v\in\mathbb F_2^m$. We show that the number of coordinates $j$ satisfying
$h_j = v$ and $ z_j=1 $ is even. Choose an index $i$ with $v_i=1$ and let $ R_i:=\mathrm{supp}(r_i). $
For $j\in R_i$, let $h_j^{(i)}\in\mathbb F_2^{m-1}$ be $h_j$ after deleting its $i$-th coordinate. Define \begin{equation}
    f_i(a):= \sum_{\substack{j\in R_i\\ h_j^{(i)}=a}}z_j \pmod 2. 
\end{equation}
Since $ |\mathrm{supp}(f_i)| \le |R_i| <2^D, $ it is enough to show that $f_i$ is orthogonal to every Boolean polynomial of degree at most $D-1$. Such polynomials are spanned by functions of the form
\begin{equation}
    P(a)=\prod_{s=1}^{k}(u^{(s)}\cdot a), \qquad 0\le k\le D-1.
\end{equation}
Set $ \alpha_1=\cdots=\alpha_{D-k}=r_i$ and
\begin{equation}
    \alpha_{D-k+s} =\sum_{\ell\ne i}u_\ell^{(s)}r_\ell, \qquad s=1,\ldots,k.
\end{equation}
By assumption,
\begin{equation}
    0 = \sum_{j=1}^n z_j(r_i)_j^{D-k} \prod_{s=1}^{k} \left( \sum_{\ell\ne i}u_\ell^{(s)}r_\ell \right)_j = \sum_{j\in R_i} z_j \prod_{s=1}^{k} \left(u^{(s)}\cdot h_j^{(i)}\right)= \sum_{a\in\mathbb F_2^{m-1}}f_i(a)P(a).
\end{equation}
Thus $f_i$ is orthogonal to every polynomial of degree at most $D-1$. If $m-1<D$, every function on $\mathbb F_2^{m-1}$ has degree at most $D-1$, and therefore $f_i=0$. If $m-1\ge D$, then Proposition~\ref{prop:reed-muller-duality} gives
\begin{equation}
    \mathrm{RM}(D-1,m-1)^\perp =\mathrm{RM}(m-1-D,m-1),
\end{equation}
and every nonzero word in this code has weight at least $2^D$. Since $|\mathrm{supp}(f_i)|<2^D$, it follows that $f_i=0$. Evaluating at the vector obtained from $v$ by deleting its $i$-th coordinate gives
\begin{equation}
     \sum_{j:h_j=v}z_j=0.
\end{equation}
Hence the coordinates with $z_j=1$ and column pattern $v$ can be paired. For each pair $(j,k)$, one has $h_j=h_k$, and therefore
\begin{equation}
    (e_j+e_k)\cdot r_a =(h_j)_a+(h_k)_a = 0 
\end{equation}
for every row $r_a$. Thus $e_j+e_k\in C_X^\perp. $ Together with the weight-one terms from the zero columns, we get \eqref{z-decomposition}.
\end{proof}

The intuition behind the following proposition is to isolate the genuinely
level-$D$ part of the diagonal gate. The parity vector $z:=q\bmod 2$ records precisely the qubits on which an odd power of $T_D$
appears. The finite-difference constraint will force these qubits to split into disjoint one- and two-qubit $Z$-type relations. On the codespace, each such relation either makes an odd $T_D$ factor act as a scalar or combines two $T_D$ factors into a single $T_{D-1}$ factor, while the even
part of $q$ is already level-$(D-1)$. Thus the entire physical representative, and hence its logical action, descends by one level of the Clifford hierarchy.

\begin{proposition}\label{prop:diagonal-descent}
Let $D\ge 3$, and let $\mathcal C_0,\mathcal C_1$ be two $[[n,k,d]]$ stabilizer codes with $d\ge 3$. Suppose $D_q\mathcal C_0=\mathcal C_1. $ If the stabilizer group $S_0$ of $\mathcal C_0$ admits a generating set of checks of weight strictly less than $2^D$, then the logical map induced by $D_q$ from $\mathcal C_0$ to $\mathcal C_1$ lies in $\mathcal C_{D-1}$.
\end{proposition}
\begin{proof}
    Since $D_q$ is diagonal, it commutes with every pure-$Z$ Pauli. Hence the two codes have the same signed pure-$Z$ stabilizers, and therefore $C_Z(S_0)=C_Z(S_1)=:C_Z.$ Moreover, for any $z\in\mathbb F_2^n$, the operator $Z(z)$ preserves $\mathcal C_0$ if and only if it preserves $\mathcal C_1$. Let $t\in \{0,1\}$. Since $Z(z)$ preserves $\mathcal C_t$ is equivalent to $z\in C_X(S_t)^\perp,$ we obtain $C_X(S_0)=C_X(S_1)=:C_X. $

Let $g_1,\ldots,g_m$ be generators of $S_0$ of weight strictly less than $2^D$, and set $r_a:=x(g_a).$ Then $r_1,\ldots,r_m$ span $C_X$ and satisfy $|r_a|<2^D.$ The quotient $C_X^\perp/C_Z$ has dimension $k$. Choose representatives $z_1,\ldots,z_k\in C_X^\perp$ whose classes form a basis of this quotient. For each $t\in\{0,1\}$, let $|\psi_t\rangle\in\mathcal C_t$ be the unique normalized state satisfying
\begin{equation}
    Z(z_a)|\psi_t\rangle=|\psi_t\rangle, \qquad a=1,\ldots,k.
\end{equation}
Since $D_q$ commutes with every $Z(z_a)$ and maps $\mathcal C_0$ to $\mathcal C_1$, there is a phase $e^{i\theta}$ such that $ D_q|\psi_0\rangle=e^{i\theta}|\psi_1\rangle. $  Adding the operators $Z(z_a)$ does not change the $X$-projection of the stabilizer group, so both $|\psi_0\rangle$ and $|\psi_1\rangle$ have $X$-projection $C_X$. Their computational-basis support is therefore an affine coset of $C_X$. Since $D_q$ is diagonal, the two supports coincide. Write this common support as $u+C_X.$ Choose a basis $a_1,\ldots,a_r$ of $C_X$ and write \begin{equation}
    x(y)=u+\sum_{i=1}^r y_i a_i.
\end{equation}
The stabilizer equations imply that the nonzero amplitudes have equal magnitude and may be written as
\begin{equation}\label{amplitude-mag-psi}
    \langle x(y)|\psi_t\rangle = 2^{-r/2}\lambda_t i^{R_t(y)},
\end{equation}
for an overall phase $\lambda_t$ such that $|\lambda_t|=1$ where $R_t:\mathbb F_2^r\rightarrow\mathbb Z/4\mathbb Z$ has the form
\begin{equation}
    R_t(y)=\ell_t(y)+2p_t(y)\pmod 4,
\end{equation}
with $\ell_t$ affine-linear and $p_t$ of Boolean degree at most two. Indeed,
if a stabilizer with $X$-part $a_i$ has $Z$-part $b_i$, its stabilizer equation gives
\begin{equation}\label{eq:Rt-first-difference}
    R_t(y+e_i)-R_t(y) = A_i+2\sum_{j=1}^r B_{ij}y_j \pmod 4.
\end{equation}
One can see this as follows. For each $i$, choose a stabilizer with $X$-part $a_i$ and
write it as $g_i=i^{\kappa_i}X(a_i)Z(b_i)$, where $\kappa_i\in\mathbb Z/4\mathbb Z$. Its action on a computational-basis state is
\begin{equation}
g_i|x\rangle =i^{\kappa_i}(-1)^{b_i\cdot x}|x+a_i\rangle.
\end{equation}
Since $x(y+e_i)=x(y)+a_i,$  the stabilizer equation $g_i|\psi_t\rangle=|\psi_t\rangle$ relates the amplitudes at $x(y)$ and $x(y+e_i)$. Comparing the coefficient of $|x(y+e_i)\rangle$ gives
\begin{align}
2^{-r/2}\lambda_t i^{R_t(y+e_i)}
&=i^{\kappa_i}(-1)^{b_i\cdot x(y)}2^{-r/2}\lambda_t i^{R_t(y)},\\i^{R_t(y+e_i)-R_t(y)}&=i^{\kappa_i+2b_i\cdot x(y)},
\end{align}
where we used $(-1)^c=i^{2c}$ for $c\in\mathbb F_2$. Therefore,
\begin{align}
R_t(y+e_i)-R_t(y)
&\equiv \kappa_i+2b_i\cdot x(y) \pmod 4\\
&\equiv \kappa_i+2b_i\cdot u +2\sum_{j=1}^r(b_i\cdot a_j)y_j \pmod 4\\
&\equiv A_i+2\sum_{j=1}^rB_{ij}y_j \pmod 4,
\end{align}
where
\begin{equation}
A_i:=\kappa_i+2b_i\cdot u\pmod 4,
\qquad B_{ij}:=b_i\cdot a_j\pmod 2.
\end{equation}

To reconstruct $R_t$ from these first differences, write
$\Delta_iR_t(y):=R_t(y+e_i)-R_t(y)$. Since changing two distinct coordinates commutes, $\Delta_j\Delta_iR_t(y)=\Delta_i\Delta_jR_t(y)$. On the other hand, applying $\Delta_j$ to \eqref{eq:Rt-first-difference}
changes only the term containing $y_j$, and therefore
\begin{equation}
    \Delta_j\Delta_iR_t(y)\equiv 2B_{ij}\pmod 4,
\qquad \Delta_i\Delta_jR_t(y)\equiv 2B_{ji}\pmod 4.
\end{equation}
It follows that $B_{ij}=B_{ji}$ for $i\ne j$. Also, changing the $i$-th coordinate twice returns to the original point, so $\Delta_iR_t(y)+\Delta_iR_t(y+e_i)=0\pmod 4.$
Substituting \eqref{eq:Rt-first-difference} gives $2A_i+2B_{ii}=0\pmod 4,$ and hence $B_{ii}\equiv A_i\pmod 2.$ Now define
\begin{equation}
    \widetilde R_t(y) := R_t(0)+\sum_{i=1}^r A_i y_i +2\sum_{1\le i<j\le r}B_{ij}y_i y_j \pmod 4.
\end{equation}
Using $B_{ij}=B_{ji}$ and $B_{ii}\equiv A_i\pmod 2$, one gets
\begin{align}
\Delta_i\widetilde R_t(y)
&\equiv A_i(1-2y_i) +2\sum_{j\ne i}B_{ij}y_j \pmod 4\\
&\equiv A_i+2\sum_{j=1}^rB_{ij}y_j \pmod 4\\
&\equiv \Delta_iR_t(y) \pmod 4.
\end{align}
Thus $R_t-\widetilde R_t$ is unchanged by flipping any coordinate.
Since the Boolean cube $\mathbb F_2^r$ is connected by such flips and
$R_t(0)=\widetilde R_t(0)$, we conclude that
$R_t=\widetilde R_t$. Therefore,

\begin{equation}
    R_t(y)=\ell_t(y)+2p_t(y)\pmod 4,
\end{equation}
where $\ell_t(y) := R_t(0)+\sum_{i=1}^rA_i y_i$ is affine-linear and $p_t(y) := \sum_{1\le i<j\le r}B_{ij}y_i y_j$ is a Boolean polynomial of degree at most two. Now define $F_t(x(y)) := 2^{D-2}R_t(y)\pmod {\mathbb Z/2^D\mathbb Z}.$
Then since
\begin{equation}
    \exp\!\left(\frac{\pi i}{2^{D-1}}F_t(x(y))\right)=\exp\!\left(\frac{\pi i}{2}R_t(y)\right)=i^{R_t(y)}.
\end{equation}
So using \eqref{amplitude-mag-psi},
\begin{equation}
    |\psi_t\rangle = 2^{-r/2}\lambda_t \sum_{x\in u+C_X} \exp\left(\frac{\pi i}{2^{D-1}}F_t(x)\right)|x\rangle.
\end{equation}
Because $R_t$ is affine-linear plus twice a quadratic Boolean polynomial,
one has
\begin{equation}
\Delta_{\alpha_1}\cdots\Delta_{\alpha_D}F_t(x)=0 \pmod{2^D}
\end{equation}
for every $\alpha_1,\ldots,\alpha_D\in C_X$. Now define
\begin{equation}
    Q(x):=\sum_{j=1}^n q_jx_j  \pmod{\mathbb Z/2^D\mathbb Z}. 
\end{equation}
Indeed, since
\begin{equation}
  D_q|x\rangle =\exp\!\left(\frac{\pi i}{2^{D-1}}Q(x)\right)|x\rangle,  
\end{equation}
comparing with the coefficient of $|x\rangle$ in $D_q|\psi_0\rangle=e^{i\theta}|\psi_1\rangle$ gives
\begin{equation}
    \exp\!\left( \frac{\pi i}{2^{D-1}}
\bigl(F_0(x)+Q(x)-F_1(x)\bigr) \right) = e^{i\theta}\frac{\lambda_1}{\lambda_0}.
\end{equation}
The right-hand side is independent of $x$. Fixing any $x_0\in u+C_X$ and setting $\kappa:=F_0(x_0)+Q(x_0)-F_1(x_0)\pmod{2^D},$
we therefore obtain, for every $x\in u+C_X$,
\begin{equation}
F_0(x)+Q(x)=F_1(x)+\kappa\pmod{2^D}.
\end{equation} 
Taking $D$ finite differences of the above in directions $\alpha_1,\ldots,\alpha_D\in C_X$, the constant $\kappa$ disappears and the $D$th differences of $F_0$ and $F_1$
vanish, so
\begin{equation}
\Delta_{\alpha_1}\cdots\Delta_{\alpha_D}Q(x)=0\pmod{2^D}.
\end{equation}
For the Boolean coordinate function $x_j$, one has $\Delta_\alpha x_j =(x_j\oplus\alpha_j)-x_j=\alpha_j(1-2x_j).$
Iterating this identity gives
\begin{equation}
\Delta_{\alpha_1}\cdots\Delta_{\alpha_D}x_j
=(-2)^{D-1}\left(\prod_{s=1}^D(\alpha_s)_j\right)(1-2x_j)\equiv 2^{D-1}\prod_{s=1}^D(\alpha_s)_j\pmod{2^D}.
\end{equation}
Since $Q(x)=\sum_{j=1}^nq_jx_j$, it follows that
\begin{equation}
    0\equiv 2^{D-1}\sum_{j=1}^nq_j
\prod_{s=1}^D(\alpha_s)_j\pmod{2^D}.
\end{equation}
Equivalently,
\begin{equation}
    \sum_{j=1}^n(q_j\bmod2) \prod_{s=1}^D(\alpha_s)_j =0\pmod2
\end{equation}
for every $\alpha_1,\ldots,\alpha_D\in C_X$. Now let $z_j:=q_j\bmod 2.$ Applying Lemma~\ref{lem:odd-phase-pairing} to the matrix with rows $r_1,\ldots,r_m$ gives a disjoint decomposition $z = \sum_\mu c_\mu$ for $c_\mu\in C_X^\perp$ and  $|c_\mu|\le 2$. Thus each $Z(c_\mu)$ preserves $\mathcal C_0$. Since $ |c_\mu|\le 2<d, $ it cannot be a nontrivial logical Pauli. Hence $ c_\mu\in C_Z.$ For each $\mu$, there is therefore an $\epsilon_\mu\in\mathbb F_2$ such that $(-1)^{\epsilon_\mu}Z(c_\mu)\in S_0. $

Every computational-basis label $x$ appearing in a state of $\mathcal C_0$ satisfies $c_\mu\cdot x=\epsilon_\mu \pmod 2. $
There are three cases. If $c_\mu=e_j$, then $x_j$ is fixed, so $T_D$ on qubit $j$ acts as a scalar on $\mathcal C_0$. If $c_\mu=e_j+e_k$ and $\epsilon_\mu=1$, then
$x_j+x_k=1$, so the two $T_D$ factors act as $e^{\pi i/2^{D-1}}.$ If $c_\mu=e_j+e_k$ and $\epsilon_\mu=0$, then $x_j=x_k$, and $e^{\frac{\pi i}{2^{D-1}}(x_j+x_k)}=e^{\frac{\pi i}{2^{D-2}}x_j}.$ Thus the two $T_D$ factors agree on $\mathcal C_0$ with a single $T_{D-1}$ factor. Since the supports of the $c_\mu$ are disjoint, there is a transversal gate $R_{\rm odd}=\bigotimes_{j=1}^n R_j,$ with $R_j\in\mathcal C_{D-1}, $ and a phase $\gamma$ such that $D_z|\phi\rangle = \gamma R_{\rm odd}|\phi\rangle $ for every $|\phi\rangle\in\mathcal C_0$. Since $z=q\bmod 2$, write $ q-z=2q', $ with $ q'\in(\mathbb Z/2^{D-1}\mathbb Z)^n. $ Then 
\begin{equation}
    D_q=D_zD_{2q'}, \qquad D_{2q'} =\bigotimes_{j=1}^nT_{D-1}^{q'_j}. 
\end{equation}
Therefore, on $\mathcal C_0$\,,
\begin{equation}
 D_q=\gamma W, \qquad W:=R_{\rm odd}D_{2q'},
\end{equation}
where $W$ is transversal and every one-qubit factor of $W$ is in $\mathcal C_{D-1}$. Since $D_q\mathcal C_0=\mathcal C_1$, $W$ also maps $\mathcal C_0$ to $\mathcal C_1$ and induces the same logical map as $D_q$ up to phase.

It remains only to note that a transversal physical level-$E$ gate mapping one stabilizer code to another induces a logical level-$E$ map. This follows by induction on $E$. It is immediate for $E=2$. For $E>2$, conjugating a physical representative of any logical Pauli by the transversal gate gives a
transversal level-$(E-1)$ operator on the target code; the induction hypothesis then shows that its logical action lies in
$\mathcal C_{E-1}$. Hence the original logical map lies in $\mathcal C_E$. Applying this observation with $E=D-1$ proves that the logical map induced by $D_q$ lies in $\mathcal C_{D-1}$.
\end{proof}

\subsection{Proof of Theorem \ref{thm:clifford-hierarchy-wb}}
\begin{proof}[Proof of Theorem \ref{thm:clifford-hierarchy-wb}]
   Suppose, for contradiction, that $S$ has a generating set whose checks all have weight strictly less than $2^D$. By Lemma~\ref{lem:one-qubit-hierarchy-normal-form}, there exist local Clifford gates $A,B$, an integer $E\ge2$, and $q\in(\mathbb Z/2^E\mathbb Z)^n$ such that
\begin{equation}
    U\Pi=e^{i\varphi}A D_q^{(E)}B\Pi.
\end{equation}
Replacing $E$ by a larger integer if necessary, we may assume $E\ge D$, since $T_E^a=T_{E'}^{\,2^{E'-E}a}$ for every $E'\ge E$.

Set $\mathcal C_0:=B\mathcal C$ and $\mathcal C_1:=A^\dagger\mathcal C$. Local Clifford gates preserve distance and Pauli weight, so $\mathcal C_0$ and $\mathcal C_1$ have distance $d\ge3$, and the stabilizer of $\mathcal C_0$ is generated by checks of weight strictly less than $2^D$. Moreover,
\begin{equation}
    D_q^{(E)}\mathcal C_0=\mathcal C_1.
\end{equation}
With $V_0:=BV$ and $V_1:=A^\dagger V$, the induced logical map is
\begin{equation}
    V_1^\dagger D_q^{(E)}V_0=e^{-i\varphi}\overline U.
\end{equation}

We now descend the physical hierarchy level. Suppose $e>D$ and $D_{q^{(e)}}^{(e)}$ maps $\mathcal C_0$ to $\mathcal C_1$ while inducing this same logical map. Since every chosen stabilizer generator has weight less than $2^D\le2^e$, Proposition~\ref{prop:diagonal-descent} applies. Its proof constructs, up to phase, a diagonal level-$(e-1)$ representative of the same map: $R_{\rm odd}$ is a product of powers of $T_{e-1}$, while $D_{2q'}^{(e)}=D_{q'}^{(e-1)}$. Hence there exists $q^{(e-1)}\in(\mathbb Z/2^{e-1}\mathbb Z)^n$ such that $D_{q^{(e-1)}}^{(e-1)}$ maps $\mathcal C_0$ to $\mathcal C_1$ and induces the same logical map.

Iterating from $e=E$ down to $e=D$ yields a diagonal physical level-$D$ representative of $e^{-i\varphi}\overline U$. Applying Proposition~\ref{prop:diagonal-descent} once more, now at level $D$, shows that this logical map lies in $\mathcal C_{D-1}$. Since global phases do not affect inclusion in Clifford-hierarchy, $\overline U\in\mathcal C_{D-1}$, contradicting the hypothesis. Therefore $\lambda_S\ge2^D.$
\end{proof}

\subsection{Proof of Theorem \ref{thm:concatenated_no_go}}
Next we prove Theorem \ref{thm:concatenated_no_go}. First we understand the precise setting. Our result considers standard concatenated codes of the form $\mathcal{C} = \mathcal{C}_{1} \circ \mathcal{C}_2 \circ \cdots \circ \mathcal{C}_r$ where $\mathcal{C}_1$ is any stabilizer code of distance $d\geq 3$ and arbitrary rate whereas all $\mathcal{C}_j$ are $[[n_j, 1, d_j]]$ codes with $d_j \geq 3$. The composition $\mathcal{C}_j \circ \mathcal{C}_{j+1}$ corresponds to mapping $n_j$ copies of a logical qubit encoded in $\mathcal{C}_{j+1}$ to physical qubits of $\mathcal{C}_j$. Each code $\mathcal{C}_j$ has an encoding isometry $V_j$ and a projector $\Pi_j$, and we use $\Tilde{\Pi}$ to denote the projector onto $\mathcal{C}$. 

Our proof proceeds in three steps. First, we show that at each level of concatenation, if a block-transversal unitary preserves the concatenated code, it also preserves the child codes. Therefore, if a transversal gate preserves the overall depth-$r$ concatenated code, it must preserve the codespace of each subtree. We then argue that if the transversal gate implements a genuine level-$D$ non-Clifford logical on the depth-$r$ code, then for some child code at each level of the concatenation tree, the restriction of the transversal gate to that code must implement at least a genuine level-$D$ non-Clifford logical. Finally, we apply Theorem \ref{thm:clifford-hierarchy-wb} to conclude that the number of physical qubits in each constituent code must be at least $2^D$ for the previous two statements to hold, which forces the concatenated code to have at least $2^{Dr}$ qubits, completing the proof.

We begin by showing that block-transversal logical gates preserve the children codes of each node in the concatenation tree. The following lemma is important for bookkeeping.

\begin{lemma} \label{lemma:frozen_qubits}
   Let $\mathcal{C}$ be an $n$-qubit stabilizer code with stabilizer group $S$. If $\epsilon P_j\in S$, where $\epsilon \in \{\pm 1\}$ and $P_j \in \{X_j, Y_j, Z_j\}$, then $\mathcal{C} = \ket{\phi_j} \otimes \mathcal{C}'$, where $\ket{\phi_j}$ is the $+1$ eigenstate of $\epsilon P_j$ and $\mathcal{C}'$ is an $n-1$-qubit stabilizer code. That is, qubit $j$ is a fixed tensor factor.
\end{lemma}
\begin{proof}
    Because every stabilizer commutes with $P_j$, every stabilizer acts on site $j$ with $P_j$ or $I_j$. Defining $S'$ to be S restricted to all sites except $j$, it then follows that $S = S' \times \langle \epsilon P_j\rangle$. Then every codeword $\ket{\psi}\in \mathcal{C}$ can be written as $\ket{\psi'} \otimes \ket{\phi_j}$ where $\ket{\psi'}$ is in $\mathcal{C}'$, the stabilizer code with stabilizer group $S'$.
\end{proof}
This lets us formally exclude codes with trivial single-qubit tensor factors. Let $\mathcal{A}$ be an $[[n,k,d]]$ stabilizer code with stabilizer group $S$ and let $F_A$ be the set of qubit coordinates such that $j\in F_A$ if $\pm P_j \in S_A$ for a nonidentity single-qubit $P_j$. We say that $F_A$ is the set of frozen coordinates, and $R_A = [n]\backslash F_A$ is the set of active coordinates. By Lemma \ref{lemma:frozen_qubits}, we can write $\mathcal{A}$ as
\begin{equation}
    \mathcal{A} = \left(\bigotimes_{j\in F_A}\ket{\phi_j}\right) \otimes \mathcal{A}_R ,
\end{equation}
where $\mathcal{A}_R$ is itself a stabilizer code. Moreover, $\mathcal{A}_R$ has the same code distance as $\mathcal{A}$, because any Pauli in the normalizer of $S_A$ must act on site $j$ as identity or $P_j$ where $P_j\in S_A$; one can always select the normalizer with $I_j$ as the logical Pauli representative. As such, every logical Pauli of $\mathcal{A}_R$ has at least the same weight as the Paulis in some logical coset of $\mathcal{A}$.

This bookkeeping allows us to factor out the trivial part of a stabilizer code. Notationally, we use $\mathcal{A}_R$ in the following discussion to refer to the core logical part of a general stabilizer code $\mathcal{A}$, and we use $V_A^{(R)}$ and $\Pi_A^{(R)}$ to refer to the encoding isometry and projector for $\mathcal{A}_R$ respectively. We refer to the set $F_A$ as the set of frozen coordinates and its complement, $R_A$, as the set of active coordinates.

\begin{lemma}\label{lemma:concatenation_step_preserved}
    Let $\mathcal{A}, \mathcal{B}$ be $[[n_A, k, d_A]]$ and  $[[n_B, 1, d_B]]$ stabilizer codes respectively, with $d_A, d_B \geq 3$,  encoding isometries $V_A, V_B$, and projectors $\Pi_A, \Pi_B$. Let $\Tilde{V} = V_B^{\otimes n} V_A$ and $\Tilde{\Pi} = V_B^{\otimes n_A} \Pi_A (V_B^\dagger)^{\otimes n_A}$. Let $\mathcal H_{\mathcal{B}}$ denote the codomain of $V_\mathcal{B}$. Suppose the block-transversal unitary $U = \bigotimes_{j=1}^{n_A}U_j$ for $U_j \in U(\mathcal H_{\mathcal{B}})$ preserves the concatenated codespace, i.e.~$U\Tilde{\Pi} U^\dagger = \Tilde{\Pi}$, and let $\overline{U} = \Tilde{V}^\dagger U \Tilde{V}$. For every $j\in R_A$, let $u_j = V_B^\dagger U_j V_B$. Then there is a real number $\theta$ such that
    \begin{equation}
        \overline{U} = e^{i\theta} (V_A^{(R)})^\dagger \left(\bigotimes_{j\in R_A} u_j\right)V_A^{(R)} \ .
    \end{equation}
\end{lemma}
\begin{proof}
   Write $\Pi_A = \frac{1}{|S_A|}\sum_{s\in S_A} s$.  First consider any frozen $j\in F_A$. Take the partial trace over $\Tr_{\backslash j}(\Pi_A)$ on all coordinates except $j$,
   \begin{equation}
       \Tr_{\backslash j }(\Pi_A) = 2^{k-n}2^{n-1}(I + \epsilon P_j) = 2^{k}\ketbra{\phi_j}{\phi_j} \ .
   \end{equation}
   Defining $\ket{\Phi_j} = V_B\ket{\phi_j}$ to be the child-encoded version of the frozen state $\ket{\phi_j}$, we then see
   \begin{equation}
    \Tr_{\backslash j}(U \Tilde{\Pi} U^\dagger) = U_j \Tr_{\backslash j}(\Tilde{\Pi}) U^\dagger =  2^k U_j \ketbra{\Phi_j}{\Phi_j} U_j^\dagger \ , 
   \end{equation}
   where we use $\Tilde{\Pi} =  V_B^{\otimes n} \Pi_A (V_B^\dagger)^{\otimes n}$ and $\Tr_{\backslash j}(V_B^{\otimes n} \Pi_A (V_B^\dagger)^{\otimes n})  = V_B \Tr_{\backslash j}(\Pi_A) V_B^\dagger$. But since $U$ preserves the concatenated code, $ \Tr_{\backslash j}(U \Tilde{\Pi} U^\dagger) = \Tr(\Tilde{\Pi})$, so we obtain $U_j \ketbra{\Phi_j}{\Phi_j} U^\dagger = \ketbra{\Phi_j}{\Phi_j}$, which is possible only if $U_j\ket{\Phi_j} = e^{i\theta_j}\ket{\Phi_j}$ for some real phase $\theta_j$. 

   Now consider active $j \notin F_A$. Once again taking the partial trace, we have 
   \begin{equation}
       \Tr_{\backslash j }(\Pi_A) = 2^{k-n}2^{n-1}I_j = 2^{k-1}I_j\ .
   \end{equation}
   The same argument now gives $U_j \Pi_B U_j^\dagger = \Pi_B$ for every active $j$, so every block factor on an active coordinate of $\mathcal{A}$ preserves the entire child codespace. Consequently there exists a $u_j\in U(2)$ such that $U_j V_B = V_B u_j$. The concatenated encoding isometry can be written as 
   \begin{equation}
       \Tilde{V} = \left(\bigotimes_{j\in F_A} \ket{\Phi_j}\right) \otimes \left( V_B^{\otimes n-|F_A|} V_A^{(R)}\right) \ .
   \end{equation}
   Applying the block-transversal $U$ gives
   \begin{align}
       U\Tilde{V} = e^{i\Theta}\left(\bigotimes_{j\in F_A} \ket{\Phi_j}\right) \otimes \left( V_B^{\otimes n-|F_A|} \left(\bigotimes_{j\in R_A} u_j\right)V_A^{(R)}\right) \ ,
   \end{align}
   where $\Theta = \sum_{j\in F_A} \theta_j$. Since $U\Tilde{V} = \Tilde{V}\overline{U}$, we obtain
   \begin{equation}
        \overline{U} = e^{i\theta} (V_A^{(R)})^\dagger \left(\bigotimes_{j\in R_A} u_j\right)V_A^{(R)} 
   \end{equation}
   as desired.
\end{proof}
Recalling that $u_j$ is simply the logical action of the $j$-th block on the codespace of the child code, the core interpretation of this lemma is that a block-transversal unitary preserving a two-level concatenated code also preserves the codespace of the child code. We will apply this recursively at every level of concatenation. 

Next we record a standard fact about transversal non-Clifford gates.
\begin{lemma} \label{lemma:physical_to_logical_nc}
    Let $\mathcal{Q}$ be any $[[n, k, d]]$ stabilizer code preserved by a gate which decomposes as $G = \bigotimes_{j=1}^{n} g_j$. Suppose every $g_j$ lies in the $D^{\rm th}$ level of the Clifford hierarchy. Then $\overline{G}$, the logical action of $G$ on $\mathcal{Q}$, also lies in the $D^{\rm th}$ level of the Clifford hierarchy.
\end{lemma}
\begin{proof}
    We proceed by induction on $D$. For $D=1$, $G$ is a Pauli operator, so if it preserves the codespace it lies in the normalizer and therefore induces a logical Pauli. Then suppose the hypothesis holds at level $D-1$. Choose any logical Pauli operator with physical Pauli representative $P$; then $GPG^\dagger = \bigotimes_{j=1}^n g_j P_j g_j^\dagger$, and every term in the tensor product is in the $D-1^{\rm th}$ level. The tensor product of $D-1$ level gates is in the $D-1$ level, but since $P$ is a logical representative, the logical action of $G$ maps a Pauli to a level-$D-1$ operator. As such $\overline{G}$ is in the $D^{\rm th}$ level.
\end{proof}
\begin{lemma} \label{lemma:concatenated_descent}
    Operating in the setting of Lemma \ref{lemma:concatenation_step_preserved} with the additional assumption that each block factor $U_j$ is a product of arbitrary single-qubit physical unitaries, suppose that the logical unitary $\overline{U}$ is genuinely in the $D^{\rm th}$ level. Then, for some $j\in R_A$, the logical unitary $u_j$ is genuinely in the $E^{\rm th}$ level for some $E\geq D$.
\end{lemma}
\begin{proof}
    Suppose for contradiction that $u_j\in\mathcal{C}_{D-1}$ for every active coordinate $j$. Then $\bigotimes_{j\in R_A}u_j\in\mathcal{C}_{D-1}$. By Lemma \ref{lemma:concatenation_step_preserved} and Lemma \ref{lemma:physical_to_logical_nc}, this implies $\overline{U}\in\mathcal{C}_{D-1}$, a contradiction. Hence, for some $j\in R_A$, one has $u_j\notin\mathcal{C}_{D-1}$. Moreover, $u_j$ is itself induced by a transversal product of single-qubit unitaries preserving the child stabilizer code. By Lemma \ref{lem:one-qubit-hierarchy-normal-form} and Lemma \ref{lemma:physical_to_logical_nc}, $u_j$ belongs to some finite level of the Clifford hierarchy. Therefore, it is genuinely level $E$ for some $E\geq D$.
\end{proof}
Now we can complete the proof.
\begin{proof}[Proof of Theorem \ref{thm:concatenated_no_go}] 
    We apply Lemma \ref{lemma:concatenated_descent} at every level of concatenation. We have that $U = \otimes_{j=1}^n U_j$ acts as a genuine $D^{\rm th}$ level logical gate $\overline{U}$ on $\mathcal{C} = \mathcal{C}_1 \circ (\mathcal{C}_2\circ...\circ \mathcal{C}_r)$. Then, upon removing any frozen coordinates of $\mathcal{C}_1$, we are left with an $[[m_1, k_1, d_1]]$ code, with $m_1 \leq n_1$, such that the restriction of $U$ to the qubit blocks implements a transversal action $\bigotimes_{j=1}^{m_1}u_j$ on $(\mathcal{C}_1)_R$, with each $u_j$ acting logically on $\mathcal{C}_2\circ\cdots\circ\mathcal{C}_r$. Moreover Lemma \ref{lemma:concatenated_descent} implies that at least one $u_j$ must act as a genuine logical $D^{\rm th}$ level gate or higher in the Clifford hierarchy. Let $E_2$ denote the lowest level of the Clifford hierarchy in which this $u_j$ acts logically, and set $E_1 = D$. Repeating this argument at each level, we find a nondecreasing sequence $D = E_1 \leq E_2 \leq \cdots \leq E_r$    such that at the $j^{\rm th}$ level of concatenation, the logical active core of code $\mathcal{C}_j$, $(\mathcal{C}_j)_R$, is a stabilizer code admitting a genuine $E_j$-level gate. By Theorem \ref{thm:clifford-hierarchy-wb}, every such code must have $\lambda_S \geq 2^{E_j}$, which forces $n_j \geq 2^{E_j}$ since the number of physical qubits must upper bound the check weight. Hence, the total number of qubits in the concatenated code is $n = \prod_j n_j \geq \prod_j 2^{E_j} \geq 2^{Dr}$. Rearranging gives $r\leq \lfloor \log_2 n/ D\rfloor$ as claimed.
\end{proof}

\section{Proofs of Theorem \ref{thm:discrete-eastin-knill} and syndrome irreducibility} \label{app:thms_2_3}
Here we prove Theorems \ref{thm:discrete-eastin-knill} and \ref{thm:lowweight-syndr-incompleteness}. We begin with the following definition central to the proofs. 
\begin{definition}[Codespace signal susceptibility]
Let $G=\sum_{j=1}^n \sigma_j,$
where each $\sigma_j$ is a nonidentity one-qubit Pauli supported on
qubit $j$, and let $\Pi$ denote a projector onto a subspace of $\mathbb{C}^{2^n}$. Define the codespace signal susceptibility by
\begin{equation}
    \Xi_G := \inf_{c\in\mathbb R} \left\| (G-c I)\Pi \right\|_{\mathrm{op}}^2.
\end{equation}
\end{definition}
This object allows us to quantify how nontrivially a signal generated by $G$ acts on the codespace of a QECC. We will show that any sensing-compatible code must have sufficiently large codespace signal susceptibility, but codes with large $\Xi_G$ must have large, irreducible nonlocality that manifests in every generating set of their stabilizer group. The following Lemma is the core mathematical principle behind this idea.

\begin{lemma}\label{susceptibility}
    Let $\mathcal C$ be an $n$-qubit stabilizer code of distance
    $d\ge 3$ with projector $\Pi$ and stabilizer group $S$, and let $G=\sum_{j=1}^n \sigma_j$. If $\Xi_G>nL$, then there exist a set of qubits $J\subseteq[n]$ with $|J|>L$ and a stabilizer $P\in S\setminus\mathcal N_{\le L}(S)$ such that
    \begin{equation}
    [\sigma_j,Q]=0 \quad \forall j\in J,\ \forall Q\in\mathcal N_{\le L}(S), \qquad \{\sigma_j,P\}=0 \quad \forall j\in J.
    \end{equation}
\end{lemma}
\begin{proof}
    Since $d\ge 3$, the Knill--Laflamme conditions give scalars $a_i$ and $b_{ij}$ such that
    \begin{equation}
        \Pi\sigma_i\Pi=a_i\Pi, \qquad \Pi\sigma_i\sigma_j\Pi=b_{ij}\Pi,
    \end{equation}
    for every $i,j\in[n]$. Set $\Gamma_{ij}:=b_{ij}-a_i a_j$. For every $c\in\mathbb R$,
    \begin{align}
        \Pi(G-c I)^2\Pi
        &=\sum_{i,j=1}^n\Pi\sigma_i\sigma_j\Pi
        -2c\sum_{i=1}^n\Pi\sigma_i\Pi
        +c^2\Pi \nonumber\\
        &=\left(\sum_{i,j=1}^n b_{ij}
        -2c\sum_{i=1}^n a_i
        +c^2\right)\Pi \nonumber\\
        &=\left[\sum_{i,j=1}^n\Gamma_{ij}
        +\left(c-\sum_{i=1}^n a_i\right)^2\right]\Pi.
    \end{align}
    Since $G-c I$ is Hermitian,
    \begin{align}
        \|(G-c I)\Pi\|_{\mathrm{op}}^2
        &=\|\Pi(G-c I)^2\Pi\|_{\mathrm{op}} =\sum_{i,j=1}^n\Gamma_{ij}
        +\left(c-\sum_{i=1}^n a_i\right)^2.
    \end{align}
    Taking the infimum over $c\in\mathbb R$ therefore gives $\Xi_G = \sum_{i,j=1}^n\Gamma_{ij}$.

    Construct a graph on the physical qubits by joining distinct vertices $i$ and $j$ whenever $\pm\sigma_i\sigma_j\in S$. We call a connected component \textit{fixed} if it contains a vertex $r$ for which $\pm\sigma_r\in S$, and \textit{active} otherwise. We first show that fixed components contribute nothing to $\Xi_G$. Let $C$ be a fixed component and choose $r\in C$ such that $\pm\sigma_r\in S$. By multiplying the edge stabilizers along a path from $r$ to any $i\in C$, we have that $\pm\sigma_i\sigma_r\in S$, and hence $\pm\sigma_i\in S$. Thus, for every $i\in C$, there is a sign $\epsilon_i\in\{\pm1\}$ such that $\sigma_i\Pi=\epsilon_i\Pi,$
    so $a_i=\epsilon_i$. It follows that, for every $j\in[n]$, 
    \begin{equation}
\Pi\sigma_i\sigma_j\Pi=\epsilon_i\Pi\sigma_j\Pi=\epsilon_i a_j\Pi=a_i a_j\Pi.
    \end{equation}
    Therefore $b_{ij}=a_i a_j$ and $\Gamma_{ij}=0$ whenever $i$ belongs to a fixed component. Hence fixed components contribute nothing to $\Xi_G$. We next control the contribution of the active components. Let $C$ be an active component and choose a reference vertex $r\in C$. By the construction of the graph, for every $i\in C$, there is a sign $\epsilon_i\in\{\pm1\}$ such that
    $\sigma_i\Pi=\epsilon_i\sigma_r\Pi.$
    Since $C$ is active, $\pm\sigma_i\notin S$ for every $i\in C$. Moreover, since $\mathrm{wt}(\sigma_i)=1<d$, $\sigma_i$ cannot be a nontrivial logical Pauli. Thus $\sigma_i\notin\mathcal N(S)$, and hence $a_i=0$
    for every $i\in C$. Consequently, for every $i,j\in C$, 
    \begin{equation}
\Pi\sigma_i\sigma_j\Pi=\epsilon_i\epsilon_j\Pi,
    \end{equation}
    and therefore $\Gamma_{ij}=\epsilon_i\epsilon_j.$
    Moreover, if $i$ and $j$ belong to distinct active components, then $\Gamma_{ij}=0$. Indeed, since $a_i=a_j=0$, the condition $\Gamma_{ij}\neq0$ would imply  $\Pi\sigma_i\sigma_j\Pi\neq0.$
    Because $\mathrm{wt}(\sigma_i\sigma_j)\le2<d$, this is possible only if $\pm\sigma_i\sigma_j\in S$, contradiction. It follows that
    \begin{equation}
    \label{eq:bound-xi-C}
    \Xi_G=\sum_{C\ {\rm active}}\left(\sum_{i\in C}\epsilon_i\right)^2
    \le\sum_{C\ {\rm active}}|C|^2.
    \end{equation}
    If every active component had size at most $L$, then
    \begin{equation}
        \Xi_G\le\sum_{C\ {\rm active}}L|C|\le nL,
    \end{equation}
    contradicting the hypothesis. Hence there exists an active component $C$ satisfying $|C|>L$. Fix an operator $Q\in\mathcal N(S)$ with weight at most $L$. Since $Q$ commutes with every edge stabilizer $\pm\sigma_i\sigma_j$, it anticommutes with $\sigma_i$ if and only if it anticommutes with $\sigma_j$. Thus, if $Q$ anticommutes with $\sigma_i$ for some $i\in C$, it anticommutes with $\sigma_j$ for every $j\in C$. This would mean
    $\mathrm{wt}(Q)\ge|C|>L,$
    a contradiction. Therefore $[\sigma_j,Q]=0$ for every $j\in C$ and every normalizer Pauli $Q$ of weight at most $L$. Hence every $\sigma_j$, $j\in C$, commutes with every element of $\mathcal N_{\le L}(S)$. Choose a reference vertex $r\in C$. Since $C$ is active, $\pm\sigma_r\notin S$. Since $\mathrm{wt}(\sigma_r)=1<d$, it cannot be a nontrivial logical Pauli, so $\sigma_r\notin\mathcal N(S)$. Therefore there exists a stabilizer $P\in S$ such that $\{\sigma_r,P\}=0.$
    For every $j\in C$, there is a sign $\epsilon_j\in\{\pm1\}$ such that $\epsilon_j\sigma_j\sigma_r\in S$. Since $P$ commutes with every stabilizer, it commutes with $\sigma_j\sigma_r$. It follows that $\{\sigma_j,P\}=0$
    for every $j\in C$. Since every element of $\mathcal N_{\le L}(S)$ commutes with every $\sigma_j$, then
    $P\notin\mathcal N_{\le L}(S).$ Taking $J=C$ proves the lemma.
\end{proof}

Through Lemma \ref{susceptibility}, we see that the familiar appearance of coherent signal accumulation due to entanglement in metrological probe states appears as an algebraic constraint on QECCs which can support sensing: the necessity of irreducible multipartite entanglement appears as essential high-weight members of the stabilizer group. Theorems \ref{thm:discrete-eastin-knill} and \ref{thm:lowweight-syndr-incompleteness} emerge essentially as corollaries of this result.

\begin{proof}[Proof of Theorem \ref{thm:discrete-eastin-knill}]
   For every $c\in\mathbb R$, Duhamel's formula gives
\begin{equation}
    s \le \left\|\overline U_\theta-e^{-i\theta c}I_L \right\|_{\mathrm{op}} = \left\|\left( U_\theta-e^{-i\theta c}I \right)\Pi \right\|_{\mathrm{op}}\le|\theta|\left\|(G-c)\Pi \right\|_{\mathrm{op}}.
\end{equation}
Taking the infimum over $c$ gives
\begin{equation}
\label{eq:lower-bound-xi}
\Xi_G \ge \frac{s^2}{|\theta|^2}\,.
\end{equation}
Set $L=\lambda_S^{\mathcal N}$. By definition,
$S\subseteq\mathcal N_{\le L}(S).$ The conclusion of Lemma~\ref{susceptibility} is therefore impossible, since it would produce a stabilizer $P\in S\setminus\mathcal N_{\le L}(S)$. By contrapositive, $\Xi_G\le nL=n\lambda_S^{\mathcal N}$. Combining this with \eqref{eq:lower-bound-xi} gives the bound
\begin{equation}
    \lambda_S^{\mathcal N} \ge\frac{s^2}{n|\theta|^2}\,.
\end{equation}
Since $\lambda_S^{\mathcal N}\le\lambda_S$, the same lower bound holds
for $\lambda_S$.
\end{proof}

\begin{proof}[Proof of Theorem \ref{thm:lowweight-syndr-incompleteness}] By \eqref{eq:lower-bound-xi}, $\Xi_G\ge\frac{s^2}{|\theta|^2}$. Hence, whenever $L<\frac{s^2}{n|\theta|^2}$, one has $\Xi_G>nL$. Lemma~\ref{susceptibility} then gives a set of qubits $J\subseteq[n]$ with $|J|>L$ and a stabilizer $P\in S\setminus\mathcal N_{\le L}(S)$ such that \begin{equation}
    [\sigma_j,Q]=0 \quad \forall j\in J,\ \forall Q\in\mathcal N_{\le L}(S), \qquad \{\sigma_j,P\}=0 \quad \forall j\in J.
\end{equation}
Because $d\ge3$, each weight-one error $\sigma_j$ is correctable. Moreover, since $P$ anticommutes with $\sigma_j$ for every $j\in J$, it acts nontrivially on every qubit in $J$. Therefore, \begin{equation}
    \mathrm{wt}(P)\ge |J|>L.
\end{equation} 
Finally, if $|\theta|\le Cn^{-\alpha}$, then \begin{equation}
    \frac{s^2}{n|\theta|^2}\ge\frac{s^2}{C^2}n^{2\alpha-1}.
\end{equation} Hence every $L=o(n^{2\alpha-1})$ satisfies $nL<\frac{s^2}{|\theta|^2}$ for sufficiently large $n$. Moreover, every generating set of $S$ must contain a stabilizer of weight greater than $L$. Indeed, otherwise every generator would belong to $\mathcal N_{\le L}(S)$, implying $S\subseteq\mathcal N_{\le L}(S)$, but $P\in S\setminus\mathcal N_{\le L}(S)$. \end{proof}

\section{Proofs of no-go theorem for beyond-SQL sensing and specialization to transversal sensing}\label{app:no-go-for-sensing}

Here we prove Theorem \ref{thm:general-tranverse-sensing-no-go_main} by establishing the stronger Theorem \ref{thm:general-tranverse-sensing-no-go}. 

\begin{center}
    \textit{Model for general sensing}
\end{center}
First, we define a general sensing protocol consisting of $m$ rounds of signal interrogation interleaved with quantum processing. Transversal sensing, as presented in the main text, would be an $m=1$ protocol because the signal interrogation is followed immediately by error correction and readout; more general protocols, however, may utilize additional quantum processing. 

In round $r$, we say that the $n$ sensing qubits evolve for time $t_r$ under 
\begin{equation}
    \mathcal{U}_{\omega t_r}(\rho)=e^{-i\omega t_r G}\rho e^{i \omega t_r G}, \quad G=\sum_{j=1}^n \sigma_j,
\end{equation}
 where each $\sigma_j$ is a nonidentity one-qubit Pauli supported on qubit $j$. Next we precisely define a noise model in which each qubit has an $n$-independent rate of signal-aligned error. For this, we consider the evolution of the sensor during interrogation to be defined by the master equation
 \begin{equation}
     \frac{d\rho}{dt} = -i\omega[G, \rho] + \gamma \sum_{j=1}^n (\sigma_j \rho\sigma_j - \rho) \ .
 \end{equation}
 Here, $\gamma$ is a positive constant controlling the instantaneous noise strength. In round $r$ the sensor is exposed for time $t_r$; the dissipative part of the above master equation then induces the single-qubit quantum channel 
 \begin{equation}
     \mathcal{N}_{\gamma, t_r, j}^{\mathrm{deph}}(\rho) = (1-p_{\gamma, r})\rho +  p_{\gamma,r}\sigma_j \rho \sigma_j\,, \qquad p_{\gamma, r} = \frac{1-e^{-2\gamma t_r}}{2}
 \end{equation}
 on every qubit $j=1...n$, yielding the composite dephasing channel $     \mathcal{N}_{\gamma, t_r}^{\mathrm{deph}}(\rho) = \bigotimes_{j=1}^n  \mathcal N_{\gamma, t_r, j}^{\mathrm{deph}}$. It then follows that the sensor experiences the quantum channel $\mathcal{S}_\omega^{(r)} = \mathcal{N}_{\gamma, t_r}^{\mathrm{deph}} \circ \mathcal U_{\omega t_r}$ in round $r$. In practice, the noise is often simple $Z$-axis dephasing, or more generally noise that acts on multiple Pauli axes; here we isolate the weakest noise model which is aligned with the signal.

 Between sensing rounds, the protocol may apply arbitrary parameter independent CPTP maps $\Phi_0, ..., \Phi_m$ to the sensing qubits, the memory, fresh ancillas, and any classical measurement registers. However, as discussed in \cite{CotlerGongKannan2026NoisyQuantumLearning}, it is important that the error associated with increased rounds of quantum processing be accounted for independently from noise during the sensing period. This is both operationally and theoretically important. In practice, quantum control cannot be executed arbitrarily fast. Moreover, if all processing were assumed error-free, the sensing model would not produce physically realistic optimality constraints because it would be strictly advantageous to take $m\rightarrow \infty$ to artificially mitigate signal dephasing errors. As such, we introduce a separate layer of errors associated with each interface between signal interrogation and quantum processing. For this, we consider an optimistic erasure channel. We say that $\tau$, an $n$-independent constant, is the minimum timescale at which quantum control can be executed, and denote the strength of the interface noise (a quantity determined by hardware) by another constant $\eta$. Then, after each layer of dephasing noise, every qubit in the sensor register is erased with probability $\geq 1-e^{-\eta \tau}$ and left intact with probability $e^{-\eta \tau}$. Defining the erasure probability $p_{\mathrm{e}} = 1-e^{-\eta \tau}$, this leaves us with the single-qubit erasure channel 
 \begin{equation}
 \mathcal{N}_{p_{\mathrm{e}}, j}^{\mathrm{e}}(\rho) = (1-p_{\mathrm{e}})\rho + p_{\mathrm{e}} \Tr(\rho) \ketbra{e}{e} \,,     
 \end{equation}
 acting identically on each $j= 1...n$ and the composite $n$-qubit channel $\mathcal N_{p_{\mathrm{e}}}^{\mathrm{e}} = \bigotimes_{j=1}^n \mathcal{N}_{p_{\mathrm{e}},j}^{\mathrm{e}}$. 

As such, the final state of the probe for the most general sensing protocol is
\begin{equation}
    \rho_\omega^{\mathrm{out}}=\Phi_m \circ \Lambda_\omega^{(m)}\circ \Phi_{m-1}\circ \cdots \circ \Phi_1 \circ \Lambda_\omega^{(1)}\circ \Phi_0(\rho_{\mathrm{in}}),
\end{equation}
where 
$\Lambda_{\omega}^{(r)}:=\mathcal{N}_{p_{\mathrm{e}}}^{\mathrm{e}}\circ\mathcal{N}_{\gamma, t_r}^{\mathrm{deph}}\circ \mathcal{U}_{\omega t_r}$ acts on the sensing qubits and trivially on the memory. We further remark that our introduction of the interface erasure channel is purely to obtain a maximally expressive lower bound on a general, realistic metrological protocol. The asymptotic degradation of beyond-SQL protocols still holds if this erasure channel is removed and all quantum processing is assumed noiseless.

\begin{center}
    \textit{Proof of no-go theorem}
\end{center}

With the model understood, we record the two metrological bounds that will be used in the proof.

\begin{lemma}[Dynamical process distinguishability]
\label{lemma:dynamical_distinguish}
Let $\rho(\omega)$ be a piecewise continuously differentiable family of quantum states. Let $L(\omega)$ be the symmetric logarithmic derivative defined implicitly by $\frac{d\rho(\omega)}{d\omega}=\frac{1}{2}\big(\rho(\omega)L(\omega)+L(\omega)\rho(\omega)\big)$, and let $F_Q(\omega)=\Tr\!\big(\rho(\omega)L(\omega)^2\big)$ be the quantum Fisher information. Then,
\begin{equation}
d_B\!\left(\rho(\omega_0),\rho(\omega_1)\right)\leq\frac{1}{2}\int_{\min\{\omega_0,\omega_1\}}^{\max\{\omega_0,\omega_1\}}\sqrt{F_Q(\omega)}\,d\omega.
\end{equation}
This result was proved in Ref.~\cite[Eq.~(4)]{TaddeiEtAl2013}.
\end{lemma}

\begin{lemma}[Adaptive channel-extension QFI bound]
\label{lemma:channel_extension_qfi}
Let $\Lambda_\omega^{(1)},\ldots,\Lambda_\omega^{(N)}$ be differentiable quantum channels used once each in a protocol with arbitrary parameter-independent CPTP maps between channel uses. For each $\ell$, choose a Kraus representation $\{K_{a,\omega}^{(\ell)}\}_a$ and define $\alpha_\ell:=\sum_a\dot K_{a,\omega}^{(\ell)\dagger}\dot K_{a,\omega}^{(\ell)}$ and $\beta_\ell:=\sum_a\dot K_{a,\omega}^{(\ell)\dagger}K_{a,\omega}^{(\ell)}$. If $\beta_\ell=0$ for every $\ell$, then the final output state satisfies
\begin{equation}
F_Q\!\left(\rho_\omega^{\mathrm{out}}\right)\leq4\sum_{\ell=1}^N\|\alpha_\ell\|_{\mathrm{op}}.
\end{equation}
\end{lemma}

\begin{proof}
Purify every parameter-independent control map and define the Stinespring isometry
\begin{equation}
W_{\ell,\omega}:=\sum_aK_{a,\omega}^{(\ell)}\otimes|a\rangle.
\end{equation}
If $V_{\ell-1}$ is the purified control preceding the $\ell$-th channel use, the purified protocol state satisfies
\begin{equation}
|\Psi_{\ell,\omega}\rangle=W_{\ell,\omega}V_{\ell-1}|\Psi_{\ell-1,\omega}\rangle.
\end{equation}
Moreover $\dot W_{\ell,\omega}^{\dagger}\dot W_{\ell,\omega}=\alpha_\ell$ and $\dot W_{\ell,\omega}^{\dagger}W_{\ell,\omega}=\beta_\ell$. Differentiating the recursion and using $W_{\ell,\omega}^{\dagger}W_{\ell,\omega}=I$ gives
\begin{align}
\|\dot\Psi_{\ell,\omega}\|^2
&=\|\dot\Psi_{\ell-1,\omega}\|^2+\langle\Psi_{\ell-1,\omega}|V_{\ell-1}^{\dagger}\alpha_\ell V_{\ell-1}|\Psi_{\ell-1,\omega}\rangle \nonumber\\
&\quad+2\operatorname{Re}\langle\dot\Psi_{\ell-1,\omega}|V_{\ell-1}^{\dagger}W_{\ell,\omega}^{\dagger}\dot W_{\ell,\omega}V_{\ell-1}|\Psi_{\ell-1,\omega}\rangle.
\end{align}
Since $\beta_\ell=0$, differentiating $W_{\ell,\omega}^{\dagger}W_{\ell,\omega}=I$ also gives $W_{\ell,\omega}^{\dagger}\dot W_{\ell,\omega}=0$. Therefore,
\begin{equation}
\|\dot\Psi_{\ell,\omega}\|^2\leq\|\dot\Psi_{\ell-1,\omega}\|^2+\|\alpha_\ell\|_{\mathrm{op}}.
\end{equation}
Iterating over all channel uses and using $F_Q(|\Psi_\omega\rangle)\leq4\|\dot\Psi_\omega\|^2$ proves the result after tracing out the purifying registers. This is the $\beta_\ell=0$ specialization of the adaptive channel-extension bound of Ref.~\cite[Eq.~(14) and Appendix~B]{KubicaDemkowicz2021}.
\end{proof}

\begin{theorem}[General sensing bound]
\label{thm:general-tranverse-sensing-no-go}
Consider the sensing protocol defined above, let $T:=\sum_{r=1}^mt_r$, and set $\epsilon:=|\omega_1-\omega_0|$. Then,
\begin{equation}
\label{eq:general-sensing-exact-bound}
d_B\!\left(\rho_{\omega_0}^{\mathrm{out}},\rho_{\omega_1}^{\mathrm{out}}\right)\leq\epsilon\min\left\{nT,\left(n\sum_{r=1}^m\frac{(1-p_{\mathrm{e}})t_r^2e^{-4\gamma t_r}}{1-(1-p_{\mathrm{e}})e^{-4\gamma t_r}}\right)^{1/2}\right\}.
\end{equation}
Consequently, a simplified but slightly looser bound is
\begin{equation}
\label{eq:general-sensing-simple-bound}
d_B\!\left(\rho_{\omega_0}^{\mathrm{out}},\rho_{\omega_1}^{\mathrm{out}}\right)\leq\epsilon\min\left\{nT,T\sqrt{\frac{n(1-p_{\mathrm{e}})}{p_{\mathrm{e}}}},\frac{1}{2}\sqrt{\frac{nT}{\gamma}}\right\}.
\end{equation}
\end{theorem}

\begin{proof}
Choose one-qubit Clifford gates $W_j$ such that $W_j\sigma_jW_j^\dagger=Z_j$, and set $W=\bigotimes_{j=1}^nW_j$. Signal-aligned dephasing and flagged erasure are covariant under this local change of basis, so conjugating the protocol by $W$ reduces the signal generator to $G=\sum_{j=1}^nZ_j$ without changing the Bures distance.

Fix a sensing round $r$. For one sensing qubit, the combined signal, dephasing, and erasure channel has Kraus operators
\begin{equation}
K_{0,\omega}^{(r)}=\sqrt{e^{-\eta \tau}(1-p_{\gamma, r})}e^{-i\omega t_rZ}, \quad K_{1,\omega}^{(r)}=\sqrt{e^{-\eta \tau}p_{\gamma, r}}Ze^{-i\omega t_rZ}, \quad K_{2,\omega}^{(r)}=\sqrt{p_{\mathrm{e}}}|e\rangle\!\langle0|, \quad K_{3,\omega}^{(r)}=\sqrt{p_{\mathrm{e}}}|e\rangle\!\langle1|.
\end{equation}
The last two Kraus operators may be chosen independently of $\omega$ because the erased branch contains no information about the input state. Consider the equivalent Kraus representation
\begin{equation}
\widetilde K_{a,\omega}^{(r)}=\sum_{b=0}^3\left(e^{-i\omega h_r}\right)_{ab}K_{b,\omega}^{(r)},
\end{equation}
where
\begin{equation}
h_r=-\frac{t_r}{1-e^{-\eta \tau-4\gamma t_r}}
\begin{pmatrix}
0 & 2\sqrt{p_{\gamma, r}(1-p_{\gamma, r})} & 0 & 0\\
2\sqrt{p_{\gamma, r}(1-p_{\gamma, r})} & 0 & 0 & 0\\
0 & 0 & e^{-\eta \tau-4\gamma t_r} & 0\\
0 & 0 & 0 & -e^{-\eta \tau-4\gamma t_r}
\end{pmatrix}.
\end{equation}
Writing
\begin{equation}
\alpha_r:=\sum_a\dot{\widetilde K}_{a,\omega}^{(r)\dagger}\dot{\widetilde K}_{a,\omega}^{(r)}, \qquad \beta_r:=\sum_a\dot{\widetilde K}_{a,\omega}^{(r)\dagger}\widetilde K_{a,\omega}^{(r)},
\end{equation}
a direct calculation gives
\begin{equation}
\beta_r=0, \qquad \alpha_r=\frac{t_r^2e^{-\eta \tau-4\gamma t_r}}{1-e^{-\eta \tau-4\gamma t_r}}I=\frac{(1-p_{\mathrm{e}})t_r^2e^{-4\gamma t_r}}{1-(1-p_{\mathrm{e}})e^{-4\gamma t_r}}I.
\end{equation}

The full protocol contains $n$ copies of this one-qubit channel in every round. Applying Lemma~\ref{lemma:channel_extension_qfi} to the resulting $nm$ channel uses gives
\begin{equation}
\label{eq:general-sensing-qfi-bound}
F_Q\!\left(\rho_\omega^{\mathrm{out}}\right)\leq4n\sum_{r=1}^m\frac{(1-p_{\mathrm{e}})t_r^2e^{-4\gamma t_r}}{1-(1-p_{\mathrm{e}})e^{-4\gamma t_r}}.
\end{equation}
The right-hand side is independent of $\omega$, so Lemma~\ref{lemma:dynamical_distinguish} gives
\begin{equation}
\label{eq:general-sensing-bures-channel-extension}
d_B\!\left(\rho_{\omega_0}^{\mathrm{out}},\rho_{\omega_1}^{\mathrm{out}}\right)\leq\epsilon\left(n\sum_{r=1}^m\frac{(1-p_{\mathrm{e}})t_r^2e^{-4\gamma t_r}}{1-(1-p_{\mathrm{e}})e^{-4\gamma t_r}}\right)^{1/2}.
\end{equation}

There is also a noise-independent coherent bound. In a purification of the complete protocol, the derivative contribution from sensing round $r$ has norm at most $t_r\|G\|_{\mathrm{op}}=nt_r$. The triangle inequality therefore gives $\|\dot\Psi_\omega\|\leq nT$, and hence
\begin{equation}
F_Q\!\left(\rho_\omega^{\mathrm{out}}\right)\leq 4 n^2T^2.
\end{equation}
Applying Lemma~\ref{lemma:dynamical_distinguish} once more gives
\begin{equation}
d_B\!\left(\rho_{\omega_0}^{\mathrm{out}},\rho_{\omega_1}^{\mathrm{out}}\right)\leq \epsilon \, n\,T.
\end{equation}
Combining this with \eqref{eq:general-sensing-bures-channel-extension} proves \eqref{eq:general-sensing-exact-bound}.

It remains to derive the simpler bound. For every $t_r\geq0$,
\begin{equation}
\frac{(1-p_{\mathrm{e}})e^{-4\gamma t_r}}{1-(1-p_{\mathrm{e}})e^{-4\gamma t_r}}\leq\frac{1-p_{\mathrm{e}}}{p_{\mathrm{e}}}.
\end{equation}
Using $\sum_{r=1}^mt_r^2\leq T^2$ in \eqref{eq:general-sensing-bures-channel-extension} therefore gives
\begin{equation}
d_B\!\left(\rho_{\omega_0}^{\mathrm{out}},\rho_{\omega_1}^{\mathrm{out}}\right)\leq\epsilon \,T\sqrt{\frac{n(1-p_{\mathrm{e}})}{p_{\mathrm{e}}}}.
\end{equation}
Moreover,
\begin{equation}
\frac{(1-p_{\mathrm{e}})e^{-4\gamma t_r}}{1-(1-p_{\mathrm{e}})e^{-4\gamma t_r}}\leq\frac{e^{-4\gamma t_r}}{1-e^{-4\gamma t_r}}=\frac{1}{e^{4\gamma t_r}-1}\leq\frac{1}{4\gamma t_r}.
\end{equation}
Substituting this into \eqref{eq:general-sensing-bures-channel-extension} and using $\sum_{r=1}^mt_r=T$ gives
\begin{equation}
d_B\!\left(\rho_{\omega_0}^{\mathrm{out}},\rho_{\omega_1}^{\mathrm{out}}\right)\leq\frac{\epsilon}{2}\sqrt{\frac{nT}{\gamma}}.
\end{equation}
Combining the three bounds proves \eqref{eq:general-sensing-simple-bound}.
\end{proof}

The preceding argument also applies when the unknown parameter multiplies a known time-dependent signal, as in AC sensing. Suppose that during round $r$, over an interval $I_r$ of duration $t_r$, the Hamiltonian is
\begin{equation}
H_\vartheta(t)=\vartheta f_r(t)G,
\end{equation}
where $f_r(t)$ is known and satisfies $|f_r(t)|\leq1$. Most generally, the protocol may apply a known modulation $y_r(t)$ satisfying $|y_r(t)|\leq1$, and we define
\begin{equation}
a_r:=\int_{I_r}y_r(t)f_r(t)\,dt.
\end{equation}
Since the Hamiltonian is proportional to $G$ at all times and the dephasing is aligned with $G$, the channel in round $r$ is
\begin{equation}
\mathcal N_{p_{\mathrm{e}}}^{\mathrm{e}}\circ\mathcal N_{\gamma,t_r}^{\mathrm{deph}}\circ\mathcal U_{\vartheta a_r}.
\end{equation}
The proof of Theorem~\ref{thm:general-tranverse-sensing-no-go} therefore applies with $a_r$ replacing $t_r$ in each derivative of the signal unitary, while the dephasing strength continues to depend on the physical exposure time $t_r$. It gives
\begin{equation}
\label{eq:known-envelope-sensing-bound}
d_B\!\left(\rho_{\vartheta_0}^{\mathrm{out}},\rho_{\vartheta_1}^{\mathrm{out}}\right)\leq|\vartheta_1-\vartheta_0|\min\left\{n\sum_{r=1}^m|a_r|,\sqrt{\frac{n(1-p_{\mathrm{e}})}{p_{\mathrm{e}}}\sum_{r=1}^ma_r^2},\frac{1}{2}\sqrt{\frac{n}{\gamma}\sum_{r=1}^m\frac{a_r^2}{t_r}}\right\}.
\end{equation}
Since $|a_r|\leq t_r$, one has
\begin{equation}
\sum_{r=1}^m|a_r|\leq T,\qquad \sum_{r=1}^ma_r^2\leq T^2,\qquad \sum_{r=1}^m\frac{a_r^2}{t_r}\leq T.
\end{equation}
Thus the simplified bound of Theorem~\ref{thm:general-tranverse-sensing-no-go_main} holds unchanged for any known bounded signal envelope.

We can use these results to formally rule out asymptotically beyond-SQL DC and AC sensing. For DC sensing under $H_\omega=\omega G$, one has $f_r(t)=1$ and $a_r=t_r$. Two parameter values separated by $\Delta\omega=\Theta(\epsilon)$ therefore accumulate a relative single-qubit angle $|\theta|=\Theta(\epsilon T)$. For any fixed interface erasure probability $p_{\mathrm{e}}\in(0,1)$, taking
\begin{equation}
T=O\!\left(\frac{1}{\epsilon \,n^\alpha}\right)
\end{equation}
gives
\begin{equation}
d_B\!\left(\rho_{\omega_0}^{\mathrm{out}},\rho_{\omega_1}^{\mathrm{out}}\right)=O\!\left(n^{1/2-\alpha}\right).
\end{equation}
Constant-bias discrimination is therefore impossible for every $\alpha>1/2$. The same conclusion holds for AC amplitude sensing with
\begin{equation}
H_B(t)=B\sin(\omega t+\phi)G,
\end{equation}
where $\omega$ and $\phi$ are known and $B$ is unknown. Applying a $\pi$ pulse at each zero crossing reverses the sign of $G$, while leaving the signal-aligned dephasing dissipator unchanged. In the corresponding toggling frame, $y(t)=\operatorname{sgn}(\sin(\omega t+\phi))$, so two amplitudes separated by $\Delta B=\Theta(\epsilon)$ generate the relative single-qubit angle
\begin{equation}
|\theta|=|\Delta B|\int_0^T|\sin(\omega t+\phi)|\,dt.
\end{equation}
Moreover,
\begin{equation}
\int_0^T|\sin(\omega t+\phi)|\,dt=\frac{2}{\pi}T+O\!\left(\frac{1}{\omega}\right),
\end{equation}
and hence $|\theta|=\Theta(\epsilon T)$ once the interrogation contains a constant or larger number of signal periods. Equation~\eqref{eq:known-envelope-sensing-bound} then gives
\begin{equation}
d_B\!\left(\rho_{B_0}^{\mathrm{out}},\rho_{B_1}^{\mathrm{out}}\right)=O\!\left(n^{1/2-\alpha}\right)
\end{equation}
whenever $T=O(1/(\epsilon n^\alpha))$. Thus constant-bias discrimination is again impossible for every $\alpha>1/2$.

Finally we give the proof of Eq.~\eqref{eq:transversal-dephasing-bound}, which rules out any asymptotic transversal sensing strategy.
\begin{proof}[Proof of Eq.~\eqref{eq:transversal-dephasing-bound}]
Fix a recovery channel $\mathcal R$ and let $\mathcal{E}$ be the encoding channel. Set 
\begin{equation}
\mathcal I:=\mathcal R\circ\mathcal N_{\gamma,T}^{\mathrm{deph}}\circ\mathcal E,\qquad \delta_{\mathcal R}:=d_{\mathrm{ch}}(\mathcal I,\mathrm{id}_L).
\end{equation}
The assumption $\inf_{z\in\mathbb C}\|\overline U_\theta-zI_L\|_{\mathrm{op}}\geq s$ implies that $\overline U_\theta$ has two eigenvalues $\lambda_0,\lambda_1$ satisfying $|\lambda_0-\lambda_1|\geq s$. For corresponding normalized eigenvectors, define $|\psi\rangle := \tfrac{1}{\sqrt{2}}(|\lambda_0\rangle + |\lambda_1\rangle)$. Then we have
\begin{equation}
\left|\langle\psi|\overline U_\theta|\psi\rangle\right|=\frac{|\lambda_0+\lambda_1|}{2}=\sqrt{1-\frac{|\lambda_0-\lambda_1|^2}{4}}\leq\sqrt{1-\frac{s^2}{4}},
\end{equation}
and therefore
\begin{equation}
d_B\!\left(|\psi\rangle\!\langle\psi|,\overline U_\theta|\psi\rangle\!\langle\psi|\overline U_\theta^\dagger\right)\geq\arcsin\frac{s}{2}\geq\frac{s}{2}.
\label{eq:ideal-logical-separation}
\end{equation}

Define
\begin{equation}
\rho_0 := \mathcal I(|\psi\rangle\!\langle\psi|),\qquad \rho_\theta:=\mathcal R\circ\mathcal N_{\gamma,T}^{\mathrm{deph}}\circ\mathcal U_\theta\circ\mathcal E(|\psi\rangle\!\langle\psi|)\,.
\end{equation}
Since $\mathcal U_\theta\circ\mathcal E=\mathcal E\circ\overline{\mathcal U}_\theta$, one has $\rho_\theta=\mathcal I\!\left(\overline U_\theta|\psi\rangle\!\langle\psi|\overline U_\theta^\dagger\right)$. By the definition of $\delta_{\mathcal R}$,
\begin{equation}
d_B(|\psi\rangle\!\langle\psi|,\rho_0)\leq\delta_{\mathcal R}\,,\qquad d_B\!\left(\overline U_\theta|\psi\rangle\!\langle\psi|\overline U_\theta^\dagger,\rho_\theta\right)\leq\delta_{\mathcal R}\,.
\end{equation}

We now apply Theorem~\ref{thm:general-tranverse-sensing-no-go} to this one-round protocol with no interface erasure, choosing $\omega_0=0$ and $\omega_1=\theta/T$. The Markovian term in Eq.~\eqref{eq:general-sensing-simple-bound} gives
\begin{equation}
d_B(\rho_0,\rho_\theta)\leq\frac{|\theta|}{2}\sqrt{\frac{n}{\gamma T}}.
\end{equation}
Combining this with Eq.~\eqref{eq:ideal-logical-separation} and the triangle inequality gives
\begin{equation}
\delta_{\mathcal R}\geq\frac{s}{4}-\frac{|\theta|}{4}\sqrt{\frac{n}{\gamma T}}.
\end{equation}
Taking the infimum over $\mathcal R$ proves Eq.~\eqref{eq:transversal-dephasing-bound}.
\end{proof}

To interpret this bound, suppose that $|\theta|=T\epsilon\leq Cn^{-\alpha}$. Then,
\begin{equation}
\delta_n(\gamma,T)\geq\frac{s}{4}-\frac{\sqrt C}{4\sqrt\gamma}\sqrt{\epsilon \,n^{1-\alpha}}.
\end{equation}
Thus, for constant $\gamma>0$ and $\epsilon = o(n^{\alpha-1})$,
\begin{equation}
\liminf_{n\rightarrow\infty}\delta_n(\gamma,T)\geq\frac{s}{4}=\Omega(1).
\end{equation}
In this high-precision regime, any transversal implementation with $|\theta|=O(n^{-\alpha})$ therefore has logical recovery error bounded away from zero. In the complementary regime, the same result appears as a metrological obstruction. With no interface erasure, Eq.~\eqref{eq:general-sensing-simple-bound} implies that constant-bias discrimination requires
\begin{equation}
T = \Omega\!\left(\frac{\gamma}{\epsilon^2 n}\right).
\end{equation}
The proposed transversal scaling $T=O(1/(\epsilon n^\alpha))$ can be asymptotically smaller than this SQL lower bound only if $\epsilon=o(\gamma n^{\alpha-1})$, precisely the regime in which the recovery error above is nonvanishing. If instead $\epsilon = \Omega(\gamma n^{\alpha-1})$, then
\begin{equation}
\frac{1}{\epsilon \,n^\alpha} = \Omega\!\left(\frac{\gamma}{\epsilon^2 n}\right),
\end{equation}
so the proposed sensing time is already no better than the temporal SQL. Thus, under signal-aligned Markovian dephasing of constant strength, a transversal protocol either has asymptotically nonvanishing logical error or fails to provide a beyond-SQL sensing advantage.